\documentclass[preprint,12pt]{elsarticle}

\usepackage{amsmath,amsthm,amssymb}
\usepackage{proof}
\usepackage{xspace}
\usepackage{fullpage}
\usepackage{url}
\usepackage{graphics} 

\long\def\ignore#1{}

\newtheorem{defn}{Definition}
\newtheorem{lemma}[defn]{Lemma}

\newtheorem{theorem}[defn]{Theorem}

\newtheorem{cor}[defn]{Corollary}
\newtheorem{notation}[defn]{Notation}

\newcommand{\lvl}{{\rm lvl}}
\newcommand{\supp}{{\rm supp}}
\newcommand{\htf}{{\rm ht}}
\newcommand{\restrict}{\uparrow}

\newcommand{\abs}[1]{\hbox{\sl abs} \; #1}
\newcommand{\tabs}[2]{\hbox{\sl abs} \; #1 \; #2}
\newcommand{\app}[2]{\hbox{\sl app} \; #1 \; #2}
\newcommand{\arr}[2]{\hbox{\sl arr} \; #1 \; #2}
\newcommand{\of}[3]{\hbox{\sl of} \; #1 \; #2 \; #3}
\newcommand{\member}[2]{\hbox{\sl member} \; #1 \; #2}

\newcommand{\cntx}[1]{\hbox{\sl cntx} \; #1}

\newcommand{\subst}[3]{\hbox{\sl subst} \; #1 \; #2 \; #3}

\newcommand{\fresh}[2]{\hbox{\sl fresh} \; #1 \; #2}
\newcommand{\spec}[3]{\hbox{\sl spec} \; #1 \; #2 \; #3}
\newcommand{\monoTy}[1]{\hbox{\sl monoTy} \; #1}
\newcommand{\polyTy}[1]{\hbox{\sl polyTy} \; #1}
\newcommand{\tup}[1]{\langle #1\rangle}

\newcommand{\new}{\reflectbox{\ensuremath{\mathsf{N}}}\xspace}

\newcommand{\parag}[1]{\smallskip\noindent{\em #1}\quad}
\newcommand{\CSNAS}{\hbox{\sl CSNAS}}
\newcommand{\CSU}{\hbox{\sl CSU}}

\newcommand{\FOL}{FO\lambda}
\newcommand{\N}{{\rm I} \! {\rm N}}
\newcommand{\FOLDN }{\FOL^{\Delta\N}}
\newcommand{\foldnb}{$FO\lambda^{\Delta\nabla}$\xspace}
\newcommand{\G}{\ensuremath{\cal G}\xspace}
\newcommand{\logic}{\G}
\newcommand{\LG}{$LG^\omega$\xspace}
\newcommand{\ie}{{\em i.e.}}
\newcommand{\eg}{{\em e.g.}}

\newcommand{\lp}{$\lambda$Prolog\xspace}

\newcommand{\defL}{\hbox{\sl def}\mathcal{L}}
\newcommand{\defR}{\hbox{\sl def}\mathcal{R}}
\newcommand{\cL}{\hbox{\sl c}\mathcal{L}}
\newcommand{\unrhdL}{\unrhd\mathcal{L}}
\newcommand{\unrhdR}{\unrhd\mathcal{R}}
\newcommand{\cut}{\hbox{\sl cut}}
\newcommand{\botL}{\bot\mathcal{L}}
\newcommand{\topR}{\top\mathcal{R}}
\newcommand{\lorL}{\lor\mathcal{L}}
\newcommand{\lorR}{\lor\mathcal{R}}
\newcommand{\landL}{\land\mathcal{L}}
\newcommand{\landR}{\land\mathcal{R}}
\newcommand{\supsetL}{\supset\!\mathcal{L}}
\newcommand{\supsetR}{\supset\!\mathcal{R}}
\newcommand{\forallL}{\forall\mathcal{L}}
\newcommand{\forallR}{\forall\mathcal{R}}
\newcommand{\nablaL}{\nabla\mathcal{L}}
\newcommand{\nablaR}{\nabla\mathcal{R}}
\newcommand{\existsL}{\exists\mathcal{L}}
\newcommand{\existsR}{\exists\mathcal{R}}

\newcommand{\IL}{\mathcal{IL}}
\newcommand{\CIR}{\mathcal{CIR}}

\newcommand{\mueq}{\stackrel{\mu}{=}}
\newcommand{\nueq}{\stackrel{\nu}{=}}

\newcommand{\cas}[1]{[\![ #1 ]\!]}
\newcommand{\enc}[1]{\ulcorner #1 \urcorner}

\journal{Information and Computation}

\begin{document}

\begin{frontmatter}

\title{Nominal Abstraction}

\author[umn]{Andrew Gacek\corref{ajg}}
\ead{agacek@cs.umn.edu}

\author[saclay]{Dale Miller}
\ead{dale.miller@inria.fr}

\author[umn]{Gopalan Nadathur}
\ead{gopalan@cs.umn.edu}

\address[umn]{
Department of Computer Science and Engineering,
University of Minnesota
}
\address[saclay]{
INRIA Saclay - \^Ile-de-France
\& LIX/\'Ecole polytechnique
}

\cortext[ajg]{
4-192 EE/CS Building,
200 Union Street SE,
Minneapolis, MN 55455
}

\begin{abstract}

\noindent Recursive relational specifications are commonly used to
describe the computational structure of formal systems. 
Recent research in proof theory has identified two features that
facilitate direct, logic-based reasoning about such descriptions: 
the interpretation of atomic judgments through recursive definitions
and an encoding of binding constructs via generic judgments.  
However, logics encompassing these two features do not currently
allow for the definition of relations that embody dynamic aspects
related to binding, a capability needed in many
reasoning tasks.
We propose a new relation between terms called {\em nominal
  abstraction} as a means for overcoming this deficiency. 
We incorporate nominal
abstraction into a rich logic also including definitions, generic
quantification, induction, and co-induction that we then prove to be
consistent.  We present examples to show that this logic can provide
elegant treatments of binding contexts that appear in many proofs,
such as those establishing properties of typing calculi and of
arbitrarily cascading substitutions that play a role in reducibility
arguments.

\end{abstract}


\begin{keyword}
generic judgments
\sep
higher-order abstract syntax
\sep
$\lambda$-tree syntax
\sep
proof search
\sep
reasoning about operational semantics
\end{keyword}

\end{frontmatter}
\newpage

\section{Introduction}
\label{sec:intro}

This paper contributes to an increasingly important approach to using
relational specifications for formalizing and reasoning about a wide
class of computational systems. This approach, whose theoretical
underpinnings are provided by recent ideas from proof theory and proof search,
has been used with success in codifying within a logical setting the
methods of structural operational semantics that are often employed in 
describing aspects such as the evaluation and type assignment
characteristics of programming languages. The main ingredients of this
approach are the use of terms to represent the syntactic objects that
are of
interest in the relevant systems and the reflection of their dynamic
aspects into judgments over such terms.

One common application of the method has utilized recursive relational
specifications or judgments over algebraic terms. We highlight three
stages of development in the kinds of judgments that have been
employed in this context, using the transition semantics for CCS as a
motivating example \cite{milner89book}:

\parag{$(1)$ Logic programming, may behavior} Logic programming
languages allow for a natural encoding and animation of
relational specifications.  For example, Horn clauses provide a simple
and immediate encoding of CCS labeled transition systems and
unification and backtracking provide a means for exploring what is
{\em reachable} from a given process.  An early system based on this
observation was Centaur \cite{borras88}, which used Prolog 
to animate the operational semantics and typing judgments of
programming languages.  Traditional logic programming is, however,
limited to describing only {\em may} behavior judgments. For example,
using it, we are not able to prove that a given CCS process $P$ {\it cannot}
make a transition. Since this negative property is logically
equivalent to proving that $P$ is bisimilar to the null process $0$,
such systems cannot also capture bisimulation.

\parag{$(2)$ Model checking, must behavior} One way to account for
must behavior is to allow for the unfolding of specifications in both
positive and negative settings.  Proof
theoretic techniques that provided for such a treatment were developed
in the early 1990s \cite{girard92mail,schroeder-Heister93lics} and
extended in subsequent work \cite{mcdowell00tcs}. In the basic form,
these techniques require an unfolding until termination,
and are therefore applicable to {\em recursive definitions} that are
{\em noetherian}. Specifications that meet this
restriction and, hence, to which this method is applicable, include 
bisimulation for finite processes and many model checking problems. As
an example, bisimulation for finite CCS can be given an immediate and
declarative treatment using these techniques \cite{mcdowell03tcs}.

\parag{$(3)$ Theorem proving, infinite behavior} Reasoning about all
members of a domain or about possibly infinite executions requires the
addition of induction and co-induction to the above framework of
recursive definitions. Incorporating induction in proof theory
goes back to Gentzen.  The work in
\cite{mcdowell00tcs,momigliano03types,tiu04phd} provides induction and
co-induction rules associated with recursive relational specifications.
In such a setting, one can prove, for example, that (strong)
bisimulation in CCS is a congruence.

The systems that are to be specified and reasoned about often involve
terms that use names and binding.  An elegant way to treat such terms
is to encode them as $\lambda$-terms and equate them using the theory
of $\alpha$, $\beta$, and $\eta$-conversion.
The three stages discussed above need to be extended to treat
representations based on such terms. The manner in which this has been
done is illustrated next using the relational specification of the 
$\pi$-calculus \cite{milner99book}.

\parag{$(1)$ Logic programming, $\lambda$-tree syntax} Higher-order
generalizations of logic programming, such as {\em higher-order
hereditary Harrop formulas} \cite{miller91apal} and the dependently
typed LF \cite{harper93jacm}, adequately capture may behavior for
terms containing bindings.  In particular, the presence of
hypothetical and universal judgments supports the $\lambda$-tree
syntax \cite{miller00cl} approach to higher-order abstract syntax
\cite{pfenning88pldi}.  The logic programming languages \lp\
\cite{nadathur88iclp} and Twelf \cite{pfenning99cade} support such
syntax representations and can be used to provide simple
specifications of, for example, reachability in the $\pi$-calculus.

\parag{$(2)$ Model checking, $\nabla$-quantification} While the
notions of universal quantification and {\em generic judgment} are
often conflated, a satisfactory treatment of must behavior requires
splitting apart these concepts.  The $\nabla$-quantifier \cite{miller05tocl}
was introduced to encode generic judgments directly.  To illustrate
the need for this split, consider the formula $\forall w.\neg (\lambda
x.x = \lambda x.w)$.  If we think of $\lambda$-terms as denoting
abstracted syntax (terms modulo $\alpha$-conversion), this formula
should be provable (variable capture is not allowed in logically sound
substitution).  On the other hand, if we think of $\lambda$-terms as
describing functions, then the equation $\lambda y.t = \lambda y.s$ is
equivalent to $\forall y. t = s$.  But then our example formula is
equivalent to $\forall w.\neg \forall x.x = w$, which should not be
provable since it is not true in a model with a single element domain.
To think of $\lambda$-terms syntactically, we treat $\lambda
y.t = \lambda y.s$ as equivalent not to $\forall y. t = s$ but,
rather, to $\nabla y. t = s$.  Our example formula then becomes equivalent
to $\forall w.\neg \nabla x.x = w$, 
which {\it is} provable \cite{miller05tocl}. Using a representation based on
this new quantifier, the $\pi$-calculus process $(\nu x).[x=w].\bar w
x$ can be proved to be bisimilar to $0$.  Bedwyr \cite{baelde07cade}
is a model checker that treats such generic judgments.

\parag{$(3)$ Theorem proving, equality of generic judgments}
When there is only finite behavior, logics for recursive definitions
do not need the cut or initial rules, and, consequently, there is no
need to know when two judgments are the same. On the other hand, the
treatment of induction and co-induction relies on the ability to  
make such identifications: \eg, when carrying out an inductive argument
over natural numbers, one must be able to recognize when the case for
$i+1$ has been reduced to the case for $i$. This identity question is
complicated by the presence of the $\nabla$-quantifier: for example,
the proof search treatment of such quantifiers involves instantiation
with generic objects whose choice of name is arbitrary and this must
be factored into assessments of equality. The \LG proof system
\cite{tiu06lfmtp} provides a way to address this issue and uses this
to support inductive reasoning over recursive definitions. 
Using \LG encodings extended with co-induction (as described in this
paper), one can prove, for instance, that (open) 
bisimulation is a $\pi$-calculus congruence. 

The key observation underlying this paper is that logics
like \LG are still missing an ingredient that is important to many
reasoning tasks.
Within
these logics, the $\nabla$-quantifier can be used to control the
structure of terms relative to the generic judgments in which they
occur. However, these logics do not possess a complementary device for
simply and precisely characterizing such structure {\em within} the
logic. 
Consider, for example, the natural way to specify typing of
$\lambda$-terms in this setting \cite{gacek08lfmtp}. The
representation of $\lambda$-terms within this approach uses 
(meta-level) abstracted variables to encode object-level bound variables
and $\nabla$-bound variables (also called here {\em nominal
  constants}) to encode object-level free variables. 
Conceptually, the type specification 
uses recursion over the representation of $\lambda$-terms,
transforming abstracted variables into nominal constants, and building
a context that associates nominal constants with types.
Now suppose that the list $[\langle x_1,t_1\rangle,
  \ldots, \langle x_n,t_n\rangle]$ represents a particular context. The
semantics of the $\nabla$-quantifier ensures that each $x_i$ in this
list is a {\em unique} nominal constant. 
This property is important to the integrity of the type
assignment. Moreover, making it explicit
can also be important to the reasoning process; for example, a proof
of the uniqueness of type assignment would draw critically on this
fact. Unfortunately, \LG and related logics do not possess a succinct and
general way to express such a property.

This paper describes a way of realizing this missing feature, thereby
yielding a logic that represents a natural endpoint to this line of
development. The 
particular means for overcoming the deficiency is a relation between
terms
called a {\it nominal abstraction}. In its essence, nominal
abstraction is an extension of the equality relation between
terms that allows for the characterization also of occurrences of
nominal constants in such terms. Combining this relation with
definitions, we will, for instance, be able to specify a property of
the form
\[
\nabla x_1\cdots\nabla x_n.\ \cntx [\tup{x_1, t_1},\ldots,\tup{x_n,
    t_n}]
\]
which effectively asserts that {\sl cntx} is true of a list of type
assignments to $n$ distinct nominal constants. By exploiting the recursive
structure of definitions, {\sl cntx} can further be defined so that the
length of the list is arbitrary.
We integrate nominal abstraction
into a broader logical context that includes also the ability to
interpret definitions inductively and co-inductively. The 
naturalness of nominal abstraction is clear from the
modular way in which we are able to define this
extended logic and to prove it consistent. We present examples of
specification and 
reasoning to bring out the usefulness of the resulting logic,
focusing especially on the capabilities resulting from nominal
abstraction.\footnote{While there might appear to be a similarity
  between nominal abstraction and ``atom-abstraction'' in {\em nominal
    logic} \cite{pitts03ic} from this discussion, these two concepts
  are technically quite different and should not be
  confused. Section~\ref{ssec:nominal} contains a comparison between
  \logic and nominal logic that should make the differences clear.}

One of the features desired for the logic presented in this paper is that it
support the $\lambda$-tree approach to the treatment of syntax. As
discussed earlier in this section, such a treatment is typically based
on permitting 
$\lambda$-bindings into terms and using universal and hypothetical
judgments in analyzing these 
terms. Hypothetical judgments force an ``open-world''
assumption; in the setting of interest, we use them to assert
new properties of the constants that are introduced in treating
bound variables. However, our desire to be able to reason
inductively about predicate definitions provides a contradictory
tension: the definition of predicates must be fixed once and for all
in order to state induction principles. This tension is relieved in
the logic we describe by disallowing hypothetical judgments and
instead using lists (as illustrated above) to implicitly encode
contexts needed in syntactic analyses. This approach is demonstrated
in greater detail through the examples in
Section~\ref{ssec:stlc-example}. 

The rest of this paper is structured as follows. We develop a logic
called \logic, that is a rather rich logic, in the next three
sections. Section~\ref{sec:logic}  presents the rules for the core
fragment of \logic that is inherited from
\LG. Section~\ref{sec:nominal-abstraction} introduces
the nominal abstraction
relation with its associated inference rules.  Finally,
Section~\ref{sec:definitions} completes the framework by adding
the mechanism of recursive definitions together with the possibility
of interpreting these inductively or co-inductively. 
A central technical result of this paper is the cut-elimination theorem
for \logic, which is presented in Section~\ref{sec:meta-theory}: an
immediate consequence of this theorem is the consistency of \logic.
Section~\ref{sec:pattern-form} introduces a more flexible and
suggestive style for recursive definitions that allows one to directly
define generic judgments: such definitions allow for the use of
``$\nabla$ in the head.''  We show that this style of definition can
be accounted for by using the nominal abstraction predicate.
Section~\ref{sec:examples} presents a collection of examples that
illustrate the expressiveness of nominal abstraction in \logic; a
reader who is interested in seeing motivating examples first might
peruse this section before digesting the detailed proofs in the
earlier sections.
Section~\ref{sec:related} compares the development in this paper
with recent
related work on specification and reasoning techniques. 

This paper extends the conference paper \cite{gacek08lics} in two
important ways.  First, nominal abstraction is used here as a more
general and modular method for obtaining the benefits of allowing
$\nabla$-quantification in the ``heads of definitions.''  Second, the
modularity provided by nominal abstraction is exploited to allow
recursive definitions to be read inductively and co-inductively.  The
logic in \cite{gacek08lics} was also called \logic: this name is reused
here for a richer logic.
The logic developed in this paper has been implemented in the Abella
system \cite{gacek08ijcar}. Abella has been used successfully in
formalizing the proofs of theorems in a number of areas
\cite{gacek08lfmtp}. 



\section{A Logic with Generic Quantification}
\label{sec:logic}

The core logic underlying \logic is obtained by extending an
intuitionistic and predicative 
subset of Church's Simple Theory of Types with a treatment of generic
judgments. The encoding of generic judgments is based on the
quantifier called $\nabla$ (pronounced nabla) introduced by Miller and
Tiu \cite{miller05tocl} and further includes the structural rules
associated with this quantifier in the logic \LG described by Tiu
\cite{tiu06lfmtp}. 
While it is 
possible to develop a classical variant of \logic as well,  we do not
follow that path here, observing simply that the choice between an
intuitionistic and a classical interpretation can lead to
interesting differences in the meaning of specifications written in
the logic. For example, it has been shown that the 
specification of bisimulation for the $\pi$-calculus within this 
logic corresponds to {\em open bisimulation} under an intuitionistic
reading and to {\em late bisimulation} under a classical reading
\cite{tiu04fguc}.

\subsection{The basic syntax}

Following Church \cite{church40}, terms are constructed from constants
and variables using abstraction and application. All
terms are assigned types using a monomorphic typing system; these
types also constrain the set of well-formed expressions in the
expected way. 
The collection of types includes $o$, a type that corresponds to 
propositions. Well-formed terms of this type are also called
formulas. We assume that $o$ does not appear in the argument
types of any nonlogical constant. 
Two terms are considered to be equal if one can be obtained from the
other by a sequence of applications of the $\alpha$-, $\beta$- and
$\eta$-conversion rules, \ie, the $\lambda$-conversion
rules. This notion of
equality is henceforth assumed implicitly wherever there is a need to
compare terms. 
Logic is 
introduced by including special constants representing the
propositional connectives $\top$, $\bot$, $\land$, $\lor$, $\supset$
and, for every type $\tau$ that does not contain $o$, the constants
$\forall_\tau$ and $\exists_\tau$ of type $(\tau \rightarrow o)
\rightarrow o$.  The binary propositional connectives are written as
usual in infix form and the expressions $\forall_\tau x. B$ and
$\exists_\tau x. B$ abbreviate the formulas $\forall_\tau \lambda x.B$
and $\exists_\tau \lambda x.B$, respectively.  Type subscripts will be
omitted from quantified formulas when they can be inferred from the
context or are not important to the discussion. We also use a
shorthand for iterated quantification: if ${\cal Q}$ is a quantifier,
the expression ${\cal Q}x_1,\ldots,x_n.P$ will abbreviate ${\cal
Q}x_1\ldots{\cal Q}x_n.P$.

The usual inference rules for the universal quantifier can be seen as
equating it to the conjunction of all of its instances: that is, this
quantifier is treated extensionally.
There are several situations  where one
wishes to treat an expression such as ``$B(x)$ holds for all
$x$'' as a statement about the existence of a uniform argument for
every instance rather than the truth of a particular property for each
instance \cite{miller05tocl}; 
such situations typically arise when one is reasoning about the
binding structure of formal objects represented using the
$\lambda${\em -tree syntax} \cite{miller00cl} version of {\em 
higher-order abstract syntax} \cite{pfenning88pldi}.
The $\nabla$-quantifier serves to encode judgments that have this kind
of a ``generic'' property associated with them. Syntactically, this
quantifier corresponds to 
including a constant $\nabla_\tau$ of type $(\tau \rightarrow o)
\rightarrow o$ for each type $\tau$ not containing $o$.\footnote{
We may choose to allow $\nabla$-quantification at fewer types in
particular applications; such a restriction may be 
  useful in adequacy arguments for reasons we discuss later.}
As with the
other quantifiers, $\nabla_\tau x. B$ abbreviates $\nabla_\tau \lambda
x. B$ and the type subscripts are often suppressed for readability.

\subsection{Generic judgments and $\nabla$-quantification}

Towards understanding the $\nabla$-quantifier, let us consider the
rule for typing abstractions in the simply-typed $\lambda$-calculus as
an example of something that we might want to encode within
\logic. This rule has the form
\begin{equation*}
\infer[x \notin \mathop{dom}(\Gamma)]
      {\Gamma \vdash (\lambda x\!:\!\alpha. t) : \alpha \to \beta}
      {\Gamma, x\!:\!\alpha \vdash t : \beta}
\end{equation*}
In the conclusion of this rule, the variable $x$ is bound and its
scope is clearly delimited by the abstraction that binds it. It
appears that $x$ is free in the premise of the rule, but it
is in fact implicitly bound over the judgment whose subcomponents,
specifically $\Gamma$, also
constrain its identity. One way to precisely encode this rule in a
meta-logic is to
introduce an explicit quantifier over $x$ in the upper
judgment; in a proof search setting, the encoding of the rule can then
be understood as one that moves a term level binding to a formula
level binding. However,
the quantifier that is used must have special properties. First, it
should enforce a property of genericity on proofs: we want the
associated typing judgment to have a derivation that is independent of
the choice of term for $x$. Second, we should be able to assume and to
use the property that instantiation terms chosen for $x$ are distinct
from other terms appearing in the judgment, in particular, in $\Gamma$.

Neither the existential nor the universal quantifier have quite the
characteristics needed for $x$ in the encoding task considered. Miller and
Tiu \cite{miller05tocl} therefore introduced the
$\nabla$-quantifier for this purpose. Using this quantifier, the
typing rule can be represented by a formula like
$\forall \Gamma, t, \alpha, \beta.
  (\nabla x. (\Gamma, x\!:\!\alpha \vdash t\, x : \beta)) \supset
  (\Gamma \vdash (\lambda x\!:\!\alpha. t\, x) : \alpha \to \beta)$
where $t$ has a higher-order type which allows its dependency on $x$
to be made explicit.
The inference rules associated with the
$\nabla$-quantifier are designed to ensure the adequacy of such an
encoding: the formula $\nabla x.F$, also called a {\em generic
  judgment}, must be established by deriving $F$ assuming $x$ to be a
completely generic variable and in deriving $\nabla x \nabla y. F$ it
is assumed that the instantiations for $x$ and $y$ are distinct. In  
the logic \logic, we shall assume two further ``structural'' properties
for the $\nabla$-quantifier which flow naturally from the application
domains of interest. First, we shall allow for $\nabla${\em
  -strengthening}, \ie, we will take $\nabla x. F$ and $F$ to be
equivalent if $x$ does not appear in $F$. Second, we shall take the
relative order of $\nabla$-quantifiers to be irrelevant, \ie, we shall
permit a $\nabla${\em -exchange} principle; the formulas $\nabla
x\nabla y. F$ and $\nabla y \nabla x.F$ will be considered to be
equivalent. These assumptions facilitate a simplification of the
inference rules, allowing us to realize generic judgments through a special
kind of constants called {\em nominal constants}.

\subsection{A sequent calculus presentation of the core logic}

\begin{figure*}[t]
\[
\begin{array}{ccc}
\infer[id]{\Sigma : \Gamma, B \vdash B'}{B \approx B'} &
\infer[\cut]{\Sigma : \Gamma, \Delta \vdash C}
            {\Sigma : \Gamma \vdash B & \Sigma : B, \Delta \vdash C} &
\infer[\cL]{\Sigma : \Gamma, B \vdash C}
           {\Sigma : \Gamma, B, B \vdash C}
\\[8pt]
\infer[\botL]{\Sigma :\Gamma,\bot \vdash C}{} &
\infer[\lorL]{\Sigma :\Gamma,B\lor D \vdash C}
                       {\Sigma :\Gamma,B\vdash C & \Sigma:\Gamma,D\vdash C} &
\infer[\lorR,i\in\{1,2\}]{\Sigma : \Gamma \vdash B_1 \lor B_2}
                                   {\Sigma : \Gamma \vdash B_i}
\\[8pt]
\infer[\topR]{\Sigma : \Gamma \vdash \top}{} &
\infer[\landL,i\in\{1,2\}]{\Sigma : \Gamma, B_1 \land B_2
  \vdash C}{\Sigma : \Gamma, B_i \vdash C} &
\infer[\landR]{\Sigma : \Gamma \vdash B \land C}{\Sigma :
  \Gamma \vdash B & \Sigma : \Gamma \vdash C}
\end{array}
\]
\[
\begin{array}{cc}
\infer[\supsetL]{\Sigma : \Gamma, B \supset D \vdash C}
      {\Sigma : \Gamma \vdash B & \Sigma : \Gamma, D \vdash C} &
\infer[\supsetR]{\Sigma : \Gamma \vdash B \supset C}
      {\Sigma : \Gamma, B \vdash C}
\\[8pt]
\infer[\forallL]{\Sigma : \Gamma, \forall_\tau x.B \vdash C}
      {\Sigma, \mathcal{K}, \mathcal{C} \vdash t : \tau &
       \Sigma : \Gamma, B[t/x] \vdash C} &
\infer[\begin{array}{l}\\\forallR,h\notin\Sigma,\\\supp(B)=\{\vec{c}\}\end{array}]
      {\Sigma : \Gamma \vdash \forall x.B}
      {\Sigma, h : \Gamma \vdash B[h\ \vec{c}/x]}
\\[8pt]
\infer[\nablaL,a\notin \supp(B)]
      {\Sigma : \Gamma, \nabla x. B \vdash C}
      {\Sigma : \Gamma, B[a/x] \vdash C} &
\infer[\nablaR,a\notin \supp(B)]
      {\Sigma : \Gamma \vdash \nabla x.B}
      {\Sigma : \Gamma \vdash B[a/x]}
\\[8pt]
\infer[\begin{array}{l}\\\existsL,h\notin\Sigma,\\\supp(B)=\{\vec{c}\}\end{array}]
      {\Sigma : \Gamma, \exists x. B \vdash C}
      {\Sigma, h : \Gamma, B[h\; \vec{c}/x] \vdash C} &
\infer[\existsR]{\Sigma : \Gamma \vdash \exists_\tau x.B}
      {\Sigma, \mathcal{K}, \mathcal{C} \vdash t:\tau &
       \Sigma : \Gamma \vdash B[t/x]}
\end{array}
\]
\caption{The core rules of \logic}
\label{fig:core-rules}
\end{figure*}

The logic \logic assumes
that the collection of constants is partitioned into the set
$\mathcal{C}$ of nominal constants and the set $\mathcal{K}$ of 
usual, non-nominal constants.
We assume the set $\mathcal{C}$ contains an infinite number of nominal
constants for each type at which $\nabla$ quantification is permitted.
We define the {\it support} of a term (or
formula), written $\supp(t)$, as the set of nominal constants
appearing in it.
A permutation of nominal constants is a type-preserving bijection $\pi$ from
$\mathcal{C}$ to $\mathcal{C}$ such that $\{ x\ |\ \pi(x) \neq x\}$ is
finite.  The application of a permutation $\pi$ to a term $t$, denoted
by $\pi.t$, is defined as follows:
\[
\begin{array}{l@{\qquad\qquad }l}
\pi.a = \pi(a), \mbox{ if $a \in \mathcal{C}$} &
\pi.c = c, \mbox{ if $c\notin \mathcal{C}$ is atomic} \\
\pi.(\lambda x.M) = \lambda x.(\pi.M) &
\pi.(M\; N) = (\pi.M)\; (\pi.N)
\end{array}
\]
We extend the notion of equality between terms to encompass also
the application of permutations to nominal constants appearing in
them. Specifically, the relation $B \approx B'$ holds if
there is a permutation $\pi$ such that $B$ $\lambda$-converts to
$\pi.B'$.  Since $\lambda$-convertibility is an equivalence
relation and permutations are invertible and composable, it follows
that $\approx$ is an equivalence relation.

The rules defining the core of \logic are presented in Figure 
\ref{fig:core-rules}. Sequents in this logic have the form $\Sigma :
\Gamma \vdash C$ where $\Gamma$ is a multiset and the signature
$\Sigma$ contains all the free variables of $\Gamma$ and $C$. In
keeping with our restriction on quantification, we assume that $o$
does not appear in the type of any variable in $\Sigma$.  The
expression $B[t/x]$ in the quantifier rules denotes the
capture-avoiding substitution of $t$ for $x$ in the formula $B$. 
In the
$\nabla\mathcal{L}$ and $\nabla\mathcal{R}$ rules, $a$ denotes a
nominal constant of an appropriate type. In the $\exists\mathcal{L}$
and $\forall\mathcal{R}$ rule we use raising \cite{miller92jsc} to
encode the dependency of the quantified variable on the support of
$B$; the expression $(h\ \vec{c})$ in which $h$ is a fresh
eigenvariable is used in these two rules to denote the
(curried) application of $h$ to the constants appearing in the
sequence $\vec{c}$. The
$\forall\mathcal{L}$ and $\exists\mathcal{R}$ rules make use of
judgments of the form $\Sigma, \mathcal{K}, \mathcal{C} \vdash t :
\tau$. These judgments enforce the requirement that the expression
$t$ instantiating the quantifier in the rule is a well-formed term of
type $\tau$ constructed from the eigenvariables in $\Sigma$ and the
constants in ${\cal K} \cup {\cal C}$.
Notice that in contrast the $\forall\mathcal{R}$ and
$\exists\mathcal{L}$ rules seem to allow for a dependency on only a
restricted set of nominal constants.  This asymmetry is not, however,
significant: a consequence of Corollary~\ref{cor:extend} in
Section~\ref{sec:meta-theory} is that the 
dependency expressed through raising in the latter rules can be
extended to any number of nominal constants that are not in the
relevant support set without affecting the provability of sequents.

Equality modulo $\lambda$-conversion is built into the rules in
Figure~\ref{fig:core-rules}, and also into later extensions of
this logic, in a fundamental way: in particular, proofs are preserved
under the replacement of formulas in sequents by ones to which they
$\lambda$-convert.  A more involved observation
is that we can replace a formula $B$ in a sequent by another formula
$B'$ such that $B\approx B'$ without affecting the provability of the
sequent or even the very structure of the proof. As a particular
example, if $a$ and $b$ are nominal constants, then the following
three sequents are all derivable: $P\ a \vdash P\ a$, $P\ b \vdash P\
b$, and $P\ a \vdash P\ b$. The last of these examples makes clear
that nominal constants represent implicit quantification whose scope
is limited to {\em individual formulas} in a sequent rather than
ranging over the entire sequent. For the core logic,
this observation follows from the form of the $id$ rule and the fact
that permutations distribute over logical structure. We shall prove
this property explicitly for the full logic in Section~\ref{sec:meta-theory}.

The treatment of $\nabla$-quantification via nominal constants also
validates the $\nabla$-exchange and $\nabla$-strengthening principles
discussed earlier. It is interesting to note that the latter principle 
implies that every type at which one is  
willing to use $\nabla$-quantification is non-empty and, in fact,
contains an unbounded number of members.  For example, the formula
$\exists_\tau x.\top$ is always provable, even if there are no closed
terms of type $\tau$ because this formula is equivalent to
$\nabla_\tau y.\exists_\tau x.\top$, which is provable.
Similarly, for any given $n\geq 1$, the following formula
is provable
\[
\exists_\tau x_1\ldots\exists_\tau x_n.
 \left[\bigwedge_{1\leq i,j\leq n, i\not= j} x_i \not= x_j\right].
\]



\section{Characterizing Occurrences of Nominal Constants}
\label{sec:nominal-abstraction}

We are interested in adding to our logic the capability of
characterizing occurrences of nominal constants within terms and also
of analyzing the structure of terms with respect to such
occurrences. For example, we may want to define a predicate called
{\em name} that holds of a term exactly when that term is a nominal
constant. Similarly, we might need to identify a binary relation
called {\em fresh} that holds between two terms just in the case that
the first term is a nominal constant that does not occur in the second
term. Towards supporting such possibilities, we define in this section
a special binary relation called {\it nominal abstraction} and then
present proof rules that incorporate an understanding of this relation
into the logic. A formalization of these ideas requires a careful
treatment of substitution. In particular, this operation must be
defined to respect the intended formula level scope of nominal
constants. We begin our discussion with an elaboration of this aspect.

\subsection{Substitutions and their interaction with nominal constants}

The following definition reiterates a common view of substitutions 
in logical contexts. 

\begin{defn}\label{subst}
A substitution is a type preserving mapping from variables
to terms that is the identity at all but a finite number of variables.
The domain of a substitution is the set of variables that are
not mapped to themselves and its range is the
set of terms resulting from applying it to the variables in its
domain.  We write a substitution as $\{t_1/x_1,\ldots,t_n/x_n\}$
where $x_1,\ldots,x_n$ is a list of variables that contains the
domain of the substitution and $t_1,\ldots,t_n$ is the value of the
map on these variables. The support of a substitution $\theta$,
written as $\supp(\theta)$, is the set of nominal constants that appear in
the range of $\theta$. The restriction of a substitution $\theta$ to
the set of variables $\Sigma$, written as $\theta \restrict
\Sigma$, is a mapping that is like $\theta$ on the variables in
$\Sigma$ and the identity everywhere else.  
\end{defn}

A substitution essentially calls for the replacement of
variables by their associated terms in any context to which it is
applied. A complicating factor is that we will want to consider
substitutions in which nominal constants appear in the terms that are
to replace  particular variables. Such a substitution will typically
be determined relative to one 
formula in a sequent but may then have to be applied to other formulas
in the same sequent. In doing this, we have to take into account the
fact that the scopes of the implicit quantifiers over nominal
constants are restricted to individual formulas. Thus, the logically
correct application of a substitution should be accompanied by a
renaming of these nominal constants in the term being substituted into so as to
ensure that they are not confused with the ones appearing in
the range of the substitution. For example, consider the formula
$p\  a\  x$ where $a$ is a nominal constant and $x$ is a
variable; this formula is intended to be equivalent to $\nabla a.p\ 
a\  x$. If we were to substitute $f\  a$ for $x$ naively into it,
we would obtain the formula $p\  a\  (f\  a))$. However, this
results in an unintended capture of a nominal constant by an
(implicit) quantifier as a result of a substitution. To carry out the
substitution in a way that avoids such capture, we should first rename
the nominal constant $a$ in $p\  a \  x$ to some other nominal
constant $b$ and then apply the substitution to produce the formula
$p\  b\  (f \  a)$. 

\begin{defn}\label{ncasubst}
The ordinary application of a substitution $\theta$ to a term $B$ is
denoted by $B[\theta]$ and corresponds to the replacement of the
variables in $B$ by the terms that $\theta$ maps them to, making sure,
as usual, to avoid accidental binding of the variables appearing in
the range of $\theta$. More precisely, if $\theta = \{t_1/x_1,\ldots,
t_n/x_n\}$, then $B[\theta]$ is the term $(\lambda x_1\ldots\lambda x_n.B)\;
t_1\; \ldots\; t_n$; this term is, of course, considered to be equal
to any other term to which it $\lambda$-converts.  By contrast, 
the {\em nominal capture avoiding application} of $\theta$ to $B$ is
written as $B\cas{\theta}$ and is defined as follows. Assuming that
$\pi$ is a permutation of nominal constants that maps those appearing
in $\supp(B)$ to ones not appearing in $\supp(\theta)$, let $B' =
\pi.B$. Then $B\cas{\theta} = B'[\theta]$.
\end{defn}

The notation $B[\theta]$ generalizes the one used in the quantifier
rules in Figure~\ref{fig:core-rules}.  This ordinary notion of
substitution is needed to define such rules and it is used in the
proof theory.  As we will see in Section~\ref{sec:meta-theory},
however, it is nominal capture avoiding substitution that is the
logically correct notion of substitution for \G since it preserves the
provability of sequents.
For this reason, when we speak of the application of a substitution in
an unqualified way, we shall mean the nominal capture avoiding form of
this notion. It is interesting to note that as the treatment of syntax
becomes richer and more abstract, the natural notions of 
equality of expressions and of substitution also change.  
When the syntax of terms is encoded as trees, term equality is
tree equality and substitution corresponds to ``grafting.''  When
syntax involves binding operators (as in first-order formulas or
$\lambda$-terms), then it is natural for equality to become
$\lambda$-convertibility and for substitutions to be
``capture-avoiding'' in the usual sense. Here, we have introduced into 
syntax the additional 
notion of nominal constants, for which we need to upgrade equality to
the $\approx$-relation and substitution to the one which avoids the
capture of nominal constants.

The definition of the nominal capture avoiding application of a
substitution is ambiguous in that we do not uniquely specify the
permutation to be used.  We resolve this ambiguity by deeming as
acceptable {\it any} permutation that avoids conflicts. As a special
instance of the lemma below, we see that for any given formula $B$ and
substitution $\theta$, 
all the possible values for $B\cas{\theta}$ are equivalent modulo the
$\approx$ relation. Moreover, as we show in
Section~\ref{sec:meta-theory}, formulas that are equivalent under 
$\approx$ are interchangeable in the contexts of proofs. 

\begin{lemma}
\label{lem:approx-cas}
If $t \approx t'$ then $t\cas{\theta} \approx t'\cas{\theta}$.
\end{lemma}
\begin{proof}
Let $t$ be $\lambda$-convertible to $\pi_1.t'$, let $t\cas{\theta} =
(\pi_2.t)[\theta]$ where $\supp(\pi_2.t) \cap \supp(\theta) =
\emptyset$, and let $t'\cas{\theta}$ be $\lambda$-convertible to
$(\pi_3.t')[\theta]$ where $\supp(\pi_3.t')\cap \supp(\theta) =
\emptyset$. Then we define a function $\pi$ partially by the following
rules:
\begin{enumerate}
\item $\pi(c) = \pi_2.\pi_1.\pi_3^{-1}(c)$ if $c \in \supp(\pi_3.t')$ and
\item $\pi(c) = c$ if $c \in \supp(\theta)$.
\end{enumerate}
Since $\supp(\pi_3.t') \cap \supp(\theta) = \emptyset$, these rules
are not contradictory, \ie, this (partial) function is well-defined.
The range of the first rule is $\supp(\pi_2.\pi_1.\pi_3^{-1}.\pi_3.t')
= \supp(\pi_2.\pi_1.t') = \supp(\pi_2.t)$ which is disjoint from the
range of the second rule, $\supp(\theta)$. Since the mapping in each
rule is determined by a permutation, these rules together define a
one-to-one partial mapping that can be extended to a bijection on
$\mathcal{C}$. We take any such extension to be the complete
definition of $\pi$ that must therefore be a permutation.

To prove that $t\cas{\theta} \approx t'\cas{\theta}$ it suffices to
show that if $t$ is $\lambda$-convertible to $\pi_1.t'$ then
$(\pi_2.t)[\theta]$ is $\lambda$-convertible to
$\pi.((\pi_3.t')[\theta])$. We will prove this by induction on the
structure of $t'$. Permutations and substitutions distribute
over the structure of terms, thus the cases for when $t'$ is an
abstraction or application follow directly from the induction
hypothesis.
If $t'$ is a nominal constant $c$ then
$(\pi_2.t)[\theta]$ must be $\lambda$-convertible to
$(\pi_2.\pi_1.c)[\theta] = \pi_2.\pi_1.c$. Also,
$\pi.((\pi_3.t')[\theta])$ must be $\lambda$-convertible to
$\pi.\pi_3.c$. Further, in this case the first rule for $\pi$ applies
which means $\pi.\pi_3.c = \pi_2.\pi_1.\pi_3^{-1}.\pi_3.c =
\pi_2.\pi_1.c$. Thus $(\pi_2.t)[\theta]$ is
$\lambda$-convertible to $\pi.((\pi_3.t')[\theta])$. Finally, suppose
$t'$ is a variable $x$. In this case $t$ must be $\lambda$-convertible
to $x$ so that we must show $x[\theta]$ $\lambda$-converts to
$\pi.(x[\theta])$. If $x$ does not have a binding in $\theta$ then
both terms are equal. Alternatively, if $x[\theta] = s$ then $\pi.s =
s$ follows from an inner induction on $s$ and the second rule for
$\pi$.
Thus $(\pi_2.t)[\theta]$ $\lambda$-converts to
$\pi.((\pi_3.t')[\theta])$, as is required.
\end{proof}

We shall need to consider the composition of substitutions later in
this section. The definition of this notion must also pay attention to
the presence of nominal constants. 

\begin{defn}\label{nascomp}
Given a substitution $\theta$ and a permutation $\pi$ of nominal
constants, let $\pi.\theta$ denote 
the substitution that is obtained by replacing each $t/x$ in $\theta$
with $(\pi.t)/x$. Given any two substitutions $\theta$ and $\rho$, let
$\theta \circ \rho$ denote the substitution that is such that
$B[\theta\circ \rho] = B[\theta][\rho]$. In this context, the {\em
  nominal capture   avoiding composition} of $\theta$ and $\rho$ is
written as $\theta\bullet\rho$ and defined as follows. Let $\pi$ be a
permutation of nominal constants such that 
$\supp(\pi.\theta)$ is disjoint from $\supp(\rho)$. Then
$\theta\bullet\rho = (\pi.\theta)\circ \rho$.
\end{defn}\label{substequiv}
The notation $\theta \circ \rho$ in the above definition represents
the usual composition of $\theta$ and $\rho$ and can, in fact, be
given in an explicit form based on these substitutions. Thus, $\theta
\bullet \rho$ can also be presented in an explicit form. Notice that our
definition of nominal capture avoiding composition is, once again,
ambiguous because it does not fix the permutation to be used,
accepting instead any one that satisfies the constraints. However, as
before, this ambiguity is harmless. To understand this, we first
extend the notion of equivalence under permutations to substitutions.
\begin{defn}
Two substitutions $\theta$ and $\rho$ are considered to be permutation
equivalent, written $\theta \approx \rho$, if and only if there is a
permutation of nominal constants $\pi$ such that $\theta =
\pi.\rho$. This notion of equivalence may also be parameterized by a
set of variables $\Sigma$ as follows: $\theta \approx_\Sigma \rho$
just in the case that $\theta \restrict \Sigma \approx \rho \restrict
\Sigma$. 
\end{defn}
It is easy to see that all possible choices for $\theta \bullet \rho$
are permutation equivalent and that if $\varphi_1 \approx \varphi_2$
then $B\cas{\varphi_1} \approx B\cas{\varphi_2}$ for any term $B$. Thus,
if our focus is on provability, the ambiguity in
Definition~\ref{nascomp} is inconsequential by a result to be
established in Section~\ref{sec:meta-theory}. As a further observation,
note that $B\cas{\theta\bullet\rho} \approx B\cas{\theta}\cas{\rho}$
for any $B$. Hence our notion of nominal capture avoiding composition
of substitutions is sensible.

The composition operation can be used to define an ordering
relation between substitutions:
\begin{defn}\label{nasordering}
Given two substitutions $\rho$ and $\theta$, we say $\rho$ is {\em
  less general than} $\theta$, denoted by $\rho \leq \theta$, if and
only if there exists a $\sigma$ such that $\rho \approx
\theta\bullet\sigma$. This relation can also be parametrized by a
set of variables: $\rho$ is less general than $\theta$
relative to $\Sigma$, written as $\rho \leq_\Sigma \theta$, if and
only if $\rho
\restrict \Sigma \leq \theta \restrict \Sigma$.
\end{defn}
The notion of generality between substitutions that is based on
nominal capture avoiding composition has a different flavor from that
based on the traditional form of substitution composition. For
example, if $a$ is a nominal constant, the substitution $\{a/x\}$ is 
strictly less general than $\{a/x, y' a/y\}$ relative to $\Sigma$ for any
$\Sigma$ which contains $x$ and $y$. To see this, note that we can
compose the latter substitution with $\{(\lambda z.y)/y'\}$ to obtain
the former, but the naive attempt to compose the former with
$\{y'a/y\}$ yields $\{b/x, y'a/y\}$ where $b$ is a nominal constant
distinct from $a$. In fact, the ``most general'' solution relative to $\Sigma$
containing $\{a/x\}$ will be $\{a/x\}\cup \{z'a/z \mid
z\in\Sigma\backslash\{x\}\}$.

\subsection{Nominal Abstraction} 

The nominal abstraction relation allows implicit formula level
bindings represented by nominal constants to be moved into explicit
abstractions over terms.  The following notation is useful for
defining this relationship.

\begin{notation}
Let $t$ be a term, let $c_1,\ldots,c_n$ be distinct nominal constants that
possibly occur in $t$, and let $y_1,\ldots,y_n$ be distinct variables
not occurring in $t$ and such that, for $1 \leq i \leq n$, $y_i$ and
$c_i$ have the same type. Then we write $\lambda c_1
\ldots\lambda c_n. t$ to denote the term $\lambda y_1 \ldots \lambda
y_n . t'$ where $t'$ is the term obtained from $t$ by replacing
$c_i$ by $y_i$ for $1\leq i\leq n$. 
\end{notation}

There is an ambiguity in the notation introduced above in that 
the choice of variables $y_1,\ldots,y_n$ is not fixed. However, this
ambiguity is harmless: the terms that are produced by acceptable
choices are all equivalent under a renaming of bound variables.

\begin{defn}\label{nominal-abstraction}
Let $n\ge 0$ and 
let $s$ and $t$ be terms of type $\tau_1 \to \cdots \to \tau_n \to
\tau$ and $\tau$, respectively; notice, in particular, that $s$ takes 
$n$ arguments to yield a term of the same type as $t$.
Then the expression $s \unrhd t$ is a formula that is referred to as a
nominal abstraction of degree $n$ or simply as a nominal abstraction. The
symbol $\unrhd$ is used here in an overloaded way in that the degree
of the nominal abstraction it participates in can vary.
The nominal abstraction $s \unrhd t$ of degree $n$ is said to hold just in
the case that $s$ $\lambda$-converts to $\lambda c_1\ldots c_n.t$ for
some nominal constants $c_1,\ldots,c_n$.
\end{defn}

Clearly, nominal abstraction of degree $0$ is the same as equality
between terms based on $\lambda$-conversion, and we will therefore use
$=$  to denote this relation in that situation. In the more general
case, 
the term on the left of the operator serves as a pattern for isolating
occurrences of nominal constants.
For example, if $p$ is a binary constructor and $c_1$ and $c_2$ are
nominal constants, then the nominal abstractions of the following
first row hold while those of the second do not.
 \begin{align*}
\lambda x. x &\unrhd c_1 &
\lambda x. p\ x\ c_2 &\unrhd p\ c_1\ c_2 &
\lambda x. \lambda y. p\ x\ y &\unrhd p\ c_1\ c_2 \\
\lambda x. x &\not\mathrel\unrhd p\ c_1\ c_2 &
\lambda x. p\ x\ c_2 &\not\mathrel\unrhd p\ c_2\ c_1 &
\lambda x. \lambda y. p\ x\ y &\not\mathrel\unrhd p\ c_1\ c_1
\end{align*}

The symbol $\unrhd$ corresponds, at the moment, to a
mathematical relation that holds between pairs of terms as 
explicated by Definition~\ref{nominal-abstraction}. We now
overload this symbol by treating it also as a binary
predicate symbol of \logic. In the next subsection we shall add
inference rules to make the mathematical understanding of 
$\unrhd$ coincide with its syntactic use as a predicate in
sequents. It is, of course, necessary to be able to determine when
we mean to use $\unrhd$ in the mathematical sense and when as a
logical symbol. When we write an expression such as $s\unrhd t$
without qualification, this should be read as a logical formula 
whereas if we say that ``$s\unrhd t$ holds'' then we are referring to
the abstract relation from 
Definition~\ref{nominal-abstraction}. We might also sometimes use an 
expression such as ``$(s\unrhd t)\cas{\theta}$
holds.'' In this case, we first treat $s \unrhd t$ as a formula to
which we apply the substitution $\theta$ in a nominal capture avoiding
way to get a (syntactic) expression of the form $s'\unrhd t'$. We then
read $\unrhd$ in the mathematical sense, interpreting the overall
expression as the assertion that ``$s'\unrhd t'$ holds.'' Note in this
context that 
$s \unrhd t$ constitutes a 
single formula when read syntactically and hence the expression
$(s\unrhd t)\cas{\theta}$ is, in general, {\it not} equivalent to 
the expression $s\cas{\theta}\unrhd t\cas{\theta}$. 

In the proof-theoretic setting, nominal abstraction will be used with
terms that contain free occurrences of variables for which
substitutions can be made. The following definition is relevant to
this situation.

\begin{defn}\label{nasolution}
A substitution $\theta$ is said to be a solution to the nominal
abstraction $s \unrhd t$ just in the case that $(s\unrhd t)\cas{\theta}$ holds.  
\end{defn}

Solutions to a nominal abstraction can be used to provide rich 
characterizations of the structures of terms. For example, consider
the nominal abstraction 
$(\lambda x.\fresh x T) \unrhd S$ in which $T$ and $S$ are
variables and {\sl fresh} is a binary predicate symbol.  Any solution
to this problem requires that $S$ be
substituted for by a term of the form $\fresh a R$ where $a$
is a nominal constant and $R$ is a term in which $a$ does not appear,
\ie, $a$ must be ``fresh'' to $R$.

An important property of solutions to a nominal abstraction is that
these are preserved under permutations to nominal constants. We
establish this fact in the lemma below; this lemma will be used later
in showing the stability of the provability of sequents with
respect to the replacement of formulas by ones they are equivalent to
modulo the $\approx$ relation.

\begin{lemma}
\label{lem:na-approx}
Suppose $(s\unrhd t) \approx (s'\unrhd t')$. Then $s\unrhd t$ and
$s'\unrhd t'$ have exactly the same solutions. In particular,
$s\unrhd t$ holds if and only if $s'\unrhd t'$ holds.
\end{lemma}
\begin{proof}
We prove the particular result first. It suffices to show it in
the forward direction since $\approx$ is symmetric. Let $\pi$ be a
permutation such that the expression $s'\unrhd t'$ $\lambda$-converts to
$\pi.(s\unrhd t)$. Now suppose $s \unrhd t$ holds since $s$
$\lambda$-converts to $\lambda\vec{c}.t$. Then an inner induction on
$t'$ shows that $s'$
$\lambda$-converts to $\lambda(\pi.\vec{c}).t'$ where $\pi.\vec{c}$ is
the result of applying $\pi$ to each element in the sequence $\vec{c}$.
Thus $s' \unrhd t'$ holds.

For the general result it again suffices to show it in one direction,
\ie, that all the solutions of $s\unrhd t$ are solutions to $s'\unrhd
t'$. Let $\theta$ be a substitution such that $(s\unrhd
t)\cas{\theta}$ holds. By Lemma~\ref{lem:approx-cas}, $(s\unrhd
t)\cas{\theta} \approx (s'\unrhd t')\cas{\theta}$. When the
substitutions are carried out, this relation has the same form as the
particular result from the first half of this proof, and thus
$(s'\unrhd t')\cas{\theta}$ holds.
\end{proof}

\subsection{Proof rules for nominal abstraction}

\begin{figure}[t]
\begin{align*}
\infer[\unrhdL]{\Sigma : \Gamma, s \unrhd t \vdash C}
{\left\{\Sigma\theta : \Gamma\cas{\theta} \vdash C\cas{\theta} \;|\;
  \hbox{$\theta$ is a solution to $(s \unrhd t)$}
  \right\}_\theta}
&&
\infer[\unrhdR,\ \hbox{$s \unrhd t$ holds}]
{\Sigma : \Gamma \vdash s \unrhd t}
{}
\end{align*}
\caption{Nominal abstraction rules}
\label{fig:na-rules}
\bigskip
\[\infer[\unrhdL_{\CSNAS}]
       {\Sigma : \Gamma, s \unrhd t \vdash C}
       {\left\{\Sigma\theta : \Gamma\cas{\theta} \vdash C\cas{\theta}
            \;|\;
          \theta \in \CSNAS(\Sigma, s, t)
        \right\}_\theta}
\]
\caption{A variant of $\unrhdL$ based on \CSNAS}
\label{fig:csnas}
\end{figure}

We now add the left and right introduction rules for $\unrhd$ 
that are shown in Figure~\ref{fig:na-rules} to link its use as a
predicate symbol to its mathematical interpretation.
The expression $\Sigma \theta$ in the $\unrhdL$ rule denotes the
application of a substitution $\theta=\{t_1/x_1,\ldots,t_n/x_n\}$ to
the signature $\Sigma$ that is defined to be the signature that
results when removing 
from $\Sigma$ the variables $\{x_1,\ldots,x_n\}$ and then adding
every variable that is free in any term in $\{t_1,\ldots,t_n\}$.
Notice also that in the same inference rule the operator 
$\cas{\theta}$ is applied to a multiset of formulas in
the natural way: 
$\Gamma\cas{\theta}=\{B\cas{\theta}\;|\; B\in\Gamma\}$.
Note that the $\unrhdL$ rule has
an {\it a priori} unspecified number of premises that 
depends on the number of substitutions that are solutions to the relevant
nominal abstraction. 
If $s \unrhd t$ expresses an unsatisfiable constraint, meaning that it
has no solutions, then the premise of $\unrhdL$ is empty and the rule
provides an immediate proof of its conclusion.

The $\unrhdL$ and $\unrhdR$ rules capture nicely the intended
interpretation of nominal abstraction. However, there is an obstacle
to using the former rule in derivations: this rule has an infinite
number of premises any time the nominal abstraction $s \unrhd t$ has a
solution. We can overcome this difficulty by describing a rule that
includes only a few of these premises but in such way that their
provability ensures the provability of all the other premises.  Since
the provability of $\Gamma \vdash C$ implies the provability of
$\Gamma\cas{\theta} \vdash C\cas{\theta}$ for any $\theta$ (a property
established formally in Section~\ref{sec:meta-theory}), if the first
sequent is a premise of an occurrence of the $\unrhdL$ rule, the
second does not need to be used as a premise of that same rule
occurrence.  Thus, we can limit the set of premises to be considered
if we can identify with any given nominal abstraction a (possibly
finite) set of solutions from which any other solution can be obtained
through composition with a suitable substitution. The following
definition formalizes the idea of such a ``covering set.'' 

\begin{defn}\label{csnas}
A {\em complete set of nominal abstraction solutions} (\CSNAS) of $s$
and $t$ on $\Sigma$
is a set $S$ of substitutions such
that
\begin{enumerate}
\item each $\theta \in S$ is a solution to $s \unrhd t$, and
\item for every solution $\rho$ to $s \unrhd t$, there exists a
$\theta \in S$ such that $\rho \leq_\Sigma \theta$.
\end{enumerate}
We denote any such set by $\CSNAS(\Sigma, s, t)$.
\end{defn}
Using this definition we present an alternative version of $\unrhdL$
in Figure~\ref{fig:csnas}. Note that if we can find a finite complete
set of nominal abstraction solutions then the number of premises to
this rule will be finite.

\begin{theorem}\label{th:csnas}
The rules $\unrhdL$ and $\unrhdL_{\CSNAS}$ are inter-admissible.
\end{theorem}
\begin{proof}
Suppose we have the following arbitrary instance of $\unrhdL$ in a
derivation: 
\begin{equation*}
\infer
 [\unrhdL]
 {\Sigma : \Gamma, s \unrhd t \vdash C}
 {\left\{
\Sigma\theta : \Gamma\cas{\theta} \vdash C\cas{\theta} \;|\;
     \hbox{$\theta$ is a solution to $(s \unrhd t)$}
\right\}_\theta}
\end{equation*}
This rule can be replaced with a use of $\unrhdL_{\CSNAS}$ instead if
we could be certain that, for each $\rho \in \CSNAS(\Sigma,s,t)$, it is
the case that $\Sigma\rho : \Gamma\cas{\rho} \vdash
C\cas{\rho}$ is included in the set of premises of the shown rule
instance. But this must be the case: by the 
definition  of $\CSNAS$, each such $\rho$ is a solution to $s \unrhd
t$.

In the other direction, suppose we have the following arbitrary
instance of $\unrhdL_{\CSNAS}$.
\begin{equation*}
\infer
 [\unrhdL_{\CSNAS}]
 {\Sigma : \Gamma, s \unrhd t \vdash C}
 {\left\{
\Sigma\theta : \Gamma\cas{\theta} \vdash C\cas{\theta}
      \;|\;
    \theta \in \CSNAS(\Sigma, s, t)
\right\}_\theta}
\end{equation*}
To replace this rule with a use of the $\unrhdL$ rule
instead, we need to be able to construct a
derivation of $\Sigma\rho : \Gamma\cas{\rho} \vdash C\cas{\rho}$ for
each $\rho$ that is a solution to $s \unrhd t$. By the definition of
$\CSNAS$, we know that for any such $\rho$ there exists a $\theta \in
\CSNAS(\Sigma,s,t)$ such that $\rho \leq_\Sigma \theta$, \ie, such
that there 
exists a $\sigma$ for which $\rho\restrict\Sigma \approx
(\theta\restrict\Sigma) \bullet \sigma$. Since we are considering the
application of these substitutions to a sequent all of whose
eigenvariables are contained in $\Sigma$, we can drop the restriction
on the substitutions and suppose that $\rho \approx \theta \bullet
\sigma$. Now, we shall show in Section~\ref{sec:meta-theory} that if a
sequent has a derivation then the result of applying a substitution to
it in a nominal capture-avoiding way produces a sequent that also has
a derivation. Using this observation, it follows that
$\Sigma\theta\sigma : \Gamma\cas{\theta}\cas{\sigma} \vdash
C\cas{\theta}\cas{\sigma}$  has a proof. But this sequent is
permutation equivalent to $\Sigma\rho : \Gamma\cas{\rho} \vdash
C\cas{\rho}$ which must, again by a result established explicitly in
Section~\ref{sec:meta-theory}, also have a proof.
\end{proof}

Theorem~\ref{th:csnas} allows us to choose which of the left rules we
wish to consider in any given context. We shall assume the $\unrhdL$
rule in the formal treatment in the rest of this paper, leaving the
use of the $\unrhdL_{\CSNAS}$ rule to practical applications of the
logic.

\subsection{Computing complete sets of nominal abstraction solutions} 

For the $\unrhdL_{CSNAS}$ rule to be useful, we need an effective way
to compute restricted complete sets of nominal abstraction
solutions. We show here that the task of finding such complete sets of
solutions can be reduced to that of finding complete sets of unifiers
(\CSU) for higher-order unification problems \cite{huet75tcs}. In the
straightforward approach to finding a solution to a nominal
abstraction $s \unrhd t$, we would first identify a substitution
$\theta$ that we apply to $s \unrhd t$ to get $s' \unrhd t'$ and we
would subsequently look for nominal constants to abstract from $t'$ to
get $s'$.  To relate this problem to the usual notion of unification,
we would like to invert this order: in particular, we would like to
consider all possible ways of abstracting over nominal constants first
and only later think of applying substitutions to make the terms
equal. The difficulty with this second approach is that we do not know
which nominal constants might appear in $t'$ until after the
substitution is applied. However, there is a way around this
problem. Given the nominal abstraction $s \unrhd t$ of degree $n$, we
first consider substitutions for the variables occurring in it that
introduce $n$ new nominal constants in a completely general way.  Then
we consider all possible ways of abstracting over the nominal
constants appearing in the altered form of $t$ and, for each of these
cases, we look for a complete set of unifiers. 

The idea described above is formalized in the following definition and
associated theorem. We use the notation $\CSU(s,t)$ in them to denote
an arbitrary but fixed selection of a complete set of unifiers for
the terms $s$ and $t$.

\begin{defn}\label{def:s}
Let $s$ and $t$ be terms of type $\tau_1 \to \ldots \to \tau_n
\to \tau$ and $\tau$, respectively. Let $c_1,\ldots,c_n$ be $n$
distinct nominal constants disjoint from $\supp(s\unrhd t)$ such that,
for $1 \leq i \leq n$, $c_i$ has the type $\tau_i$. Let $\Sigma$ be a
set of variables and for each $h \in \Sigma$ of type $\tau'$, let
$h'$ be a distinct variable not in $\Sigma$ that has type
$\tau_1\to \ldots\to\tau_n\to \tau'$. Let $\sigma = \{h'\ c_1\ \ldots\
c_n/h \mid h \in \Sigma\}$ and let $s' = s[\sigma]$ and $t' =
t[\sigma]$. Let
\begin{equation*}
C = \bigcup_{\vec{a}} \CSU(\lambda \vec{b}.s', \lambda \vec{b}. \lambda\vec{a}.t')
\end{equation*}
where $\vec{a}= a_1,\ldots,a_n$ ranges over all selections of $n$
distinct nominal constants from $\supp(t)\cup \{\vec{c}\}$ such that,
for $1 \leq i \leq n$,
$a_i$ has type $\tau_i$ and $\vec{b}$ is some corresponding listing of
all the nominal constants in $s'$ and $t'$ that are not included in
$\vec{a}$. 
Then we define
\begin{equation*}
S(\Sigma, s, t) = \{ \sigma \bullet \rho \mid \rho \in C \}
\end{equation*}
\end{defn}

The use of the substitution $\sigma$ above represents
another instance of the application of the general technique of
raising that allows  
certain variables (the $h$ variables in this definition) whose
substitution instances might depend on certain nominal constants
($c_1,\ldots,c_n$ here) to be replaced by new variables of higher type
(the $h'$ variables) whose substitution instances are not allowed to 
depend on those nominal constants. This technique was previously used
in the $\existsL$ and $\forallR$ rules presented in
Section~\ref{sec:logic}.  

An important observation concerning Definition~\ref{def:s} is that it
requires us to consider all possible (ordered) selections
$a_1,\ldots,a_n$ of distinct nominal constants from $\supp(t)\cup
\{\vec{c}\}$. The set of such selections is potentially large,
having in it at least $n!$ members. However, in the uses that
we have seen of \logic in reasoning tasks, $n$ is typically small,
often either $1$ or $2$. Moreover, in these reasoning applications,
the cardinality of the set $\supp(t)\cup \{\vec{c}\}$ is also usually
small.  

\begin{theorem}
$S(\Sigma, s, t)$ is a complete set of nominal abstraction solutions
for $s\unrhd t$ on $\Sigma$.
\end{theorem}
\begin{proof}
First note that $\supp(\sigma) \cap \supp(s\unrhd t) = \emptyset$ and
thus $(s\unrhd t)\cas{\theta}$ is equal to $(s' \unrhd t')$.
Now we must show that every element of $S(\Sigma, s, t)$ is a
solution to $s \unrhd t$. Let $\sigma\bullet\rho \in S(\Sigma, s, t)$
be an arbitrary element where $\sigma$ is as in Definition~\ref{def:s},
$\rho$ is from $\CSU(\lambda\vec{b}.s',
\lambda\vec{b}.\lambda\vec{a}.t')$, and $s' = s[\sigma]$ and $t' =
t[\sigma]$.  By the definition of $\CSU$ we know $(\lambda\vec{b}.s' =
\lambda\vec{b}.\lambda\vec{a}.t')[\rho]$. This means $(s' =
\lambda\vec{a}.t')\cas{\rho}$ holds and thus $(s' \unrhd
t')\cas{\rho}$ holds. Rewriting $s'$ and $t'$ in terms of $s$ and $t$ this
means $(s \unrhd t)\cas{\sigma}\cas{\rho}$. Thus $\sigma\bullet\rho$
is a solution to $s\unrhd t$.

In the other direction, we must show that if $\theta$ is a solution to $s
\unrhd t$ then there exists $\sigma\bullet\rho \in S(\Sigma, s, t)$
such that $\theta \le_\Sigma \sigma\bullet\rho$. Let $\theta$ be a
solution to $s\unrhd t$. Then we know $(s\unrhd t)\cas{\theta}$ holds.
The substitution $\theta$ may introduce some nominal constants which
are abstracted out of the right-hand side when determining equality, so
let us call these the {\em important} nominal constants. Let $\sigma =
\{h'\ c_1\ \ldots\ c_n/h \mid h \in \Sigma\}$ be as in
Definition~\ref{def:s} and let $\pi'$ be a permutation which maps the
important nominal constants of $\theta$ to nominal constants from
$c_1, \ldots, c_n$. This is possible since $n$ nominal constants are
abstract from the right-hand side and thus there are at most $n$
important nominal constants. Then let $\theta' = \pi'.\theta$, so that
$(s\unrhd t)\cas{\theta'}$ holds and it suffices to show that $\theta'
\le_\Sigma \sigma\bullet\rho$. Note that all we have done at this
point is to rename the important nominal constants of $\theta$ so that
they match those introduced by $\sigma$. Now we define $\rho' = \{
\lambda c_1\ldots\lambda c_n.r / h' \mid r / h \in \theta'\}$ so that
$\theta' = \sigma\bullet\rho'$. Thus $(s\unrhd
t)\cas{\sigma}\cas{\rho'}$ holds. By construction, $\sigma$ shares no
nominal constants with $s$ and $t$, thus we know $(s'\unrhd
t')\cas{\rho'}$ where $s' = s[\sigma]$ and $t' = t[\sigma]$. Also by
construction, $\rho'$ contains no important nominal constants and
thus $(s' = \lambda\vec{a}. t')\cas{\rho}$ holds for some nominal
constants $\vec{a}$ taken from $\supp(t) \cup \{\vec{c}\}$. If we let
$\vec{b}$ be a listing of all nominal constants in $s'$ and $t'$ but
not in $\vec{a}$, then $(\lambda\vec{b}. s' =
\lambda\vec{b}.\lambda\vec{a}. t')\cas{\rho}$ holds. At this point the inner
equality has no nominal constants and thus the substitution $\rho$ can
be applied without renaming: $(\lambda\vec{b}. s' =
\lambda\vec{b}.\lambda\vec{a}. t')[\rho']$ holds. By the definition of
$\CSU$, there must be a $\rho \in \CSU(\lambda\vec{b}.s',
\lambda\vec{b}.\lambda\vec{a}.t')$ such that $\rho' \le \rho$. Thus
$\sigma\bullet\rho' \le_\Sigma \sigma\bullet\rho$ as desired.
\end{proof}



\section{Definitions, Induction, and Co-induction}
\label{sec:definitions}

\begin{figure}[t]
\begin{center}
$\infer[\defL]
      {\Sigma : \Gamma, p\ \vec{t} \vdash C}
      {\Sigma : \Gamma, B\ p\ \vec{t} \vdash C}
\hspace{1in}
\infer[\defR]
      {\Sigma : \Gamma \vdash p\ \vec{t}}
      {\Sigma : \Gamma \vdash B\ p\ \vec{t}}
$
\end{center}
\caption{Introduction rules for atoms whose predicate is defined as $\forall
  \vec{x}.~p\ \vec{x} \triangleq B\ p\ \vec{x}$}
\label{fig:defrules}
\end{figure}

The sequent calculus rules presented in Figure~\ref{fig:core-rules}
treat atomic judgments as fixed, unanalyzed objects.
We now add the capability of defining such judgments by means of 
formulas, possibly involving other predicates. In particular, we shall
assume that we are given a 
fixed, finite set of \emph{clauses} of the
form $\forall \vec{x}.~p\ \vec{x} \triangleq B\ p\ \vec{x}$ where $p$
is a predicate constant that takes a number of arguments equal to the
length of $\vec{x}$. Such a clause is said to define $p$ and the
entire collection of clauses is called a {\em
  definition}. The expression $B$, called the {\em body} of the
clause, must be a term that does not contain $p$ or
any of the variables in $\vec{x}$ and must have a type such that 
$B\ p\ \vec{x}$ has type $o$.  Definitions are also restricted so that
a predicate is defined by at most one clause.
The intended interpretation of a clause $\forall \vec{x}.~p\ \vec{x}
\triangleq B\ p\ \vec{x}$ is that the atomic
formula $p\ \vec{t}$, where $\vec{t}$ is a list of terms of the same
length and type as the variables in $\vec{x}$, is true if and only if
$B\ p\ \vec{t}$ is true. 
This interpretation is realized by adding to the calculus the rules
$\defL$ and $\defR$ shown in Figure~\ref{fig:defrules} for unfolding
predicates on the left and the right of sequents using their defining
clauses.

Definitions can have a recursive structure. In particular, the
predicate $p$ can appear free in the body $B\ p\ \vec{x}$ of a clause
of the form $\forall \vec{x}.~p\ \vec{x} \triangleq B\ p\ \vec{x}$. 
A fixed-point interpretation is intended for definitions with clauses
that are recursive in this way. Additional restrictions  
are needed to ensure that fixed points actually exist in this setting
and that their use is compatible with the embedding logic. 
Two particular constraints suffice for this purpose.
First, the body of a
clause must not contain any nominal constants. This restriction
can be justified from another perspective as well: as we see in
Section~\ref{sec:meta-theory}, it helps in establishing that $\approx$
is a provability preserving equivalence between formulas. Second, definitions
should be {\em stratified} so that clauses, such as $a\triangleq
(a\supset \bot)$, in which a predicate has a negative dependency on
itself, are forbidden.  While such stratification can be enforced in
different ways, we use a simple approach to doing this in this
paper. This approach is based on associating with each predicate $p$
a natural number that is called its {\em level} and that is denoted
by $\lvl(p)$.  This measure is then extended to arbitrary formulas by
the following definition.
\begin{defn}
Given an assignment of levels to predicates, the function $\lvl$ is
extended to all formulas in $\lambda$-normal form as follows:
\begin{enumerate}
\item $\lvl(p\ \bar{t}) = \lvl(p)$
\item $\lvl(\bot) = \lvl(\top) = \lvl(s\unrhd t) = 0$
\item $\lvl(B\land C) = \lvl(B\lor C) = \max(\lvl(B),\lvl(C))$
\item $\lvl(B\supset C) = \max(\lvl(B)+1,\lvl(C))$
\item $\lvl(\forall x.B) = \lvl(\nabla x.B) = \lvl(\exists x.B) =
\lvl(B)$
\end{enumerate}
In general, the level of a formula $B$, written as $\lvl(B)$, is the
level of its $\lambda$-normal form.
\end{defn}

A definition is {\em stratified} if we can assign levels to predicates
in such a way that $\lvl(B\ p\ \vec{x}) \leq \lvl(p)$ for each clause
$\forall \vec{x}.~p\ \vec{x} \triangleq B\ p\ \vec{x}$ in that
definition.

\begin{figure}[t]
\begin{center}
$\infer[\IL]
      {\Sigma : \Gamma, p\; \vec{t} \vdash C}
      {\vec{x} : B\; S\; \vec{x} \vdash S\; \vec{x} \qquad
       \Sigma : \Gamma, S\; \vec{t} \vdash C}$\\[5pt]
provided $p$ is defined as $\forall  \vec{x}.~ p\  \vec{x}
\mueq  B\ p\  \vec{x}$ and $S$ is a term that has the same type as $p$
and does not contain nominal constants
      \\[15pt]
$\infer[\CIR]
      {\Sigma : \Gamma \vdash p\; \vec{t}}
      {\Sigma : \Gamma \vdash S\; \vec{t} \qquad
       \vec{x} : S\; \vec{x} \vdash B\; S\; \vec{x}}
$\\[5pt]
provided $p$ is defined as $\forall  \vec{x}.~ p\  \vec{x}
\nueq  B\ p\  \vec{x}$ and $S$ is a term that has the same type as $p$
and does not contain nominal constants
\end{center}
\caption{The induction left and co-induction right rules}
\label{fig:indandcoind}
\end{figure}

The $\defL$ and $\defR$ rules do not discriminate between any of the
fixed points of a definition. 
We now allow for the selection of least and greatest fixed
points so as to support inductive and co-inductive definitions of
predicates. Specifically, we
denote an inductive clause by $\forall \vec{x}.~ p\ \vec{x} \mueq
B\ p\ \vec{x}$ and a co-inductive one by $\forall \vec{x}.~ p\ \vec{x}
\nueq B\ p\ \vec{x}$.  As a refinement of the earlier restriction on
definitions, a predicate may have at most one defining clause that is
designated to be inductive, co-inductive or neither. The $\defL$ and
$\defR$ rules may be used with clauses in any one of these
forms. Clauses that are inductive admit additionally the left rule
$\IL$ shown in Figure~\ref{fig:indandcoind}. This rule is based on the
observation that the least fixed point of a monotone operator is the
intersection of all its pre-fixed points; intuitively, anything that
follows from any pre-fixed point should then also follow from the
least fixed point. In a proof search setting, the term corresponding
to the schema variable $S$ in this rule functions like the induction
hypothesis and is accordingly called the invariant of the
induction. Clauses that are co-inductive, on the other hand, admit the
right rule $\CIR$ also presented in Figure~\ref{fig:indandcoind}. This
rule reflects the fact that the greatest fixed point of a monotone
operator is the union of all the post-fixed points; any member of such
a post-fixed point must therefore also be a member of the greatest
fixed point. The substitution that is used for $S$ in this rule is
called the co-invariant or the simulation of the co-induction. Just
like the restriction on the body of clauses, in both $\IL$ and $\CIR$,
the (co-)invariant $S$ must not contain any nominal constants.

As a simple illustration of the use of these rules, consider the
clause $p \mueq p$. The desired
inductive reading of this clause
implies that $p$ must be false. In a proof-theoretic setting, we would
therefore expect that the sequent $\cdot : p \vdash \bot$ can be
proved. This can, in fact, be done by using $\IL$ with the invariant
$S = \bot$. On the other hand, consider the clause $q
\nueq q$. The co-inductive reading intended here implies that $q$ must
be true. The logic \logic satisfies this expectation: the 
sequent $\cdot : \cdot \vdash q$ can be proved using $\CIR$ with the
co-invariant $S = \top$.

The addition of inductive and co-inductive forms of clauses and
the mixing of these forms in one setting requires a stronger
stratification condition to guarantee consistency.
One condition that suffices and that is also practically acceptable is the
following that is taken from \cite{tiu09corr}: in a clause of any of the
forms $\forall \vec{x}.~ p\ \vec{x} 
\triangleq B\ p\ \vec{x}$, $\forall \vec{x}.~ p\ \vec{x} \mueq B\ p\
\vec{x}$ or $\forall \vec{x}.~ p\ \vec{x} \nueq B\ p\ \vec{x}$,
it must be that $\lvl(B\ (\lambda\vec{x}.\top)\ \vec{x}) < \lvl(p)$. This
disallows any mutual recursion between clauses, a restriction
which can easily be overcome by merging mutually recursive clauses
into a single clause.
We henceforth assume that all
definitions satisfy all three conditions described for them in this
section.
Corollary \ref{consistency} in
Section~\ref{sec:meta-theory} establishes the consistency of the logic 
under these restrictions.



\section{Some Properties of the Logic}
\label{sec:meta-theory}

We have now described the logic \logic completely: in
particular, its proof rules consist of the ones in
Figures~\ref{fig:core-rules}, \ref{fig:na-rules}, \ref{fig:defrules}
and \ref{fig:indandcoind}. This logic combines and extends the
features of several logics such as $\FOLDN$ \cite{mcdowell00tcs},
\foldnb \cite{miller05tocl}, $LG^\omega$ \cite{tiu08lgext} and
Linc$^-$  \cite{tiu09corr}. The relationship to Linc$^-$ is
of special interest to us below: \logic is a conservative
extension to this logic that is obtained by adding a treatment of the
$\nabla$ quantifier and the associated nominal constants and by
generalizing the proof rules pertaining to equality to ones dealing
with nominal abstraction. This correspondence will allow the proof of
the critical meta-theoretic property of cut-elimination for
Linc$^-$ to be lifted to \logic.

We shall actually establish three main properties of \logic in this
section.  First, we shall show that the provability of a sequent is
unaffected by the application of permutations of nominal constants to
formulas in the sequent.  This property consolidates our understanding
that nominal constants are quantified implicitly at the formula level;
such quantification also renders irrelevant the particular names chosen
for such constants. Second, we show that the application of substitution
in a nominal capture-avoiding way preserves provability; by contrast,
ordinary application of substitution does not have this property.
Finally, we show that the $\cut$ rule can be
dispensed with from the logic without changing the set of provable
sequents. This implies that the left and right rules of the logic are
balanced and, moreover, that the logic is consistent. This is the main
result of this section and its proof uses the earlier two results
together with the argument for cut-elimination for Linc$^-$.

Several of our arguments will be based on induction on the heights
of proofs. This measure is defined formally below. Notice that
the height of a proof can be an infinite ordinal because the $\unrhdL$
rule can have an infinite number of premises. Thus, we will be using 
a transfinite form of induction.

\begin{defn}
\label{def:ht}
The {\em height} of a derivation $\Pi$, denoted by $\htf(\Pi)$, is $1$
if $\Pi$ has no premise derivations and is the least upper bound of
$\{\htf(\Pi_i)+1\}_{i\in\mathcal{I}}$ if $\Pi$ has the premise
derivations $\{\Pi_i\}_{i\in\mathcal{I}}$ where $\mathcal{I}$ is some
index set. Note that the typing derivations in the rules $\forallL$
and $\existsR$ are not considered premise derivations in this sense.
\end{defn}

Many proof systems, such as Linc$^-$, include a weakening rule that
allows formulas to be dropped (reading proofs bottom-up) from the
left-hand sides of sequents. 
While \logic does not include such a rule directly, its effect is
captured in a strong sense as we show in the lemma below. Two proofs
are to be understood here and elsewhere as having the same structure
if they are isomorphic as trees, if the same rules appear at
corresponding places within them and if these rules pertain to
formulas that can be obtained one from the other via a renaming of
eigenvariables and nominal constants.

\begin{lemma}
\label{lem:proof-weak}
Let $\Pi$ be a proof of $\Sigma : \Gamma \vdash B$ and let $\Delta$ be
a multiset of formulas whose eigenvariables are contained in $\Sigma$.
Then there exists a proof of $\Sigma : \Delta, \Gamma \vdash B$ which
has the same structure as $\Pi$. In particular $\htf(\Pi) =
\htf(\Pi')$ and $\Pi$ and $\Pi'$ end with the same rule application.
\end{lemma}
\begin{proof}
The lemma can be proved by an easy induction on $\htf(\Pi)$. We omit
the details.
\end{proof}

The following lemma shows a strong form of the preservation of
provability under permutations of nominal constants appearing in
formulas, the first of our mentioned results. 

\begin{lemma}
\label{lem:proof-perm}
Let $\Pi$ be a proof of $\Sigma : B_1, \ldots, B_n \vdash B_0$ and let
$B_i \approx B_i'$ for $i \in \{0,1,\ldots,n\}$. Then there exists a
proof $\Pi'$ of $\Sigma : B_1', \ldots, B_n' \vdash B_0'$ which has
the same structure as $\Pi$. In particular $\htf(\Pi) = \htf(\Pi')$
and $\Pi$ and $\Pi'$ end with the same rule application.
\end{lemma}
\begin{proof}
The proof is by induction on $\htf(\Pi)$ and proceeds specifically by
considering the last rule used in $\Pi$. When this is a left rule, we
shall assume without loss of generality that it operates on $B_n$.

The argument is easy to provide when the last rule in $\Pi$ is one of $\botL$
or $\topR$. If this rule is an $id$, \ie, if $\Pi$ is of the form
\begin{equation*}
\infer[id]{\Sigma : B_1, \ldots, B_n \vdash B_0}
             {B_j \approx B_0}
\end{equation*}
then, since $\approx$ is an equivalence relation, it must be the case
that $B_j' \approx B_0'$. Thus, we can let
$\Pi'$ be the derivation 
\begin{equation*}
\infer[id]{\Sigma : B_1', \ldots, B_n' \vdash B_0'}
          {B_j' \approx B_0'}
\end{equation*}
If the last rule is a $\unrhdR$ applied to a nominal abstraction $s
\unrhd t$ then the result follows immediately from
Lemma~\ref{lem:na-approx}.

In the remaining cases 
we shall show that the last rule in $\Pi$ can
also have $\Sigma : B_1', \ldots, B_n' \vdash B_0'$ as a conclusion
with the premises in this application of the rule being related via
permutations in the way required by the lemma to the premises of the
rule application in $\Pi$. The lemma then follows from the induction
hypothesis. 

In the case when the last rule in $\Pi$ pertains to a binary
connective---\ie, when the rule is one of $\lorL$, $\lorR$, $\landL$,
$\landR$, $\supsetL$ or $\supsetR$---the desired conclusion follows
naturally from the observation that permutations distribute over the
connective. The proof can be similarly completed when a
$\existsL$, $\existsR$, $\forallL$ or $\forallR$ rule ends the
derivation, once we have noted that the application of permutations can
be moved under the $\exists$ and $\forall$ quantifiers. For the
$\cut$ and $\cL$ rules, we have to show that permutations
can be extended to include the newly introduced formula in the upper
sequent(s). This is easy: for the $\cut$ rule we use the identity
permutation and for $\cL$ we replicate the permutation used to obtain
$B_n'$ from $B_n$.
 
The two remaining rules from the core logic are $\nablaL$ and
$\nablaR$. The argument in these cases are similar and we consider
only the later in detail. In this case, the last rule in $\Pi$ is of
the form 
\begin{equation*}
\infer[\nablaR]
      {\Sigma : B_1, \ldots, B_n \vdash \nabla x. C}
      {\Sigma : B_1, \ldots, B_n \vdash C[a/x]}
\end{equation*}
where $a
\notin \supp(C)$. Obviously, $B_0' = \nabla x.C'$ for some $C'$ such
that $C \approx C'$. Let $d$ be a nominal constant such that $d
\notin 
\supp(C)$ and $d \notin \supp(C')$. Such a constant must exist since
both sets are finite. Then $C[a/x] \approx C[d/x] \approx C'[d/x]$. 
Thus the following
\begin{equation*}
\infer[\nablaR]
      {\Sigma : B_1', \ldots, B_n' \vdash \nabla x. C'}
      {\Sigma : B_1', \ldots, B_n' \vdash C'[d/x]}
\end{equation*}
is also an instance of the $\nablaR$ rule and its upper sequent has
the desired form.

When the last rule in $\Pi$ is $\unrhdL$, it has has the structure
\begin{equation*}
\infer[\unrhdL]
      {\Sigma : B_1, \ldots, s \unrhd t \vdash B_0}
      {\left\{ \Sigma\theta : B_1\cas{\theta}, \ldots,
                 B_{n-1}\cas{\theta} \vdash B_0\cas{\theta} \;|\;
                 \theta\ \mbox{is a solution to}\ s \unrhd t
       \right\}}
\end{equation*}
Here we know that $B_n'$ is a nominal abstraction $s'\unrhd t'$ that,
by Lemma~\ref{lem:na-approx}, has the same solutions as $s \unrhd
t$. Further, by Lemma~\ref{lem:approx-cas}, $B_i\cas{\theta} \approx
B_i'\cas{\theta}$ for any substitution $\theta$. Thus 
\begin{equation*}
\infer[\unrhdL]
      {\Sigma : B_1', \ldots, s'\unrhd t' \vdash B_0'}
      {\left\{ \Sigma\theta : B_1'\cas{\theta}, \ldots,
                B_{n-1}'\cas{\theta} \vdash
                B_0'\cas{\theta} \;|\;
          \theta\ \mbox{is a solution to}\ s' \unrhd t' 
       \right\}}
\end{equation*}
is also an instance of the $\unrhdL$ rule and its upper sequents have
the required property.

The arguments for the rules $\defL$ and $\defR$ are similar and we
therefore only consider the case for the former rule in detail. Here,
$B_n$ must be of the form $p\; \vec{t}$ where $p$ is a predicate
symbol and the upper sequent must be identical to the lower one except
for the fact that $B_n$ is replaced by a formula of the form $B\ p\;
\vec{t}$ where $B$ contains no nominal constants. Further, $B_n'$ is
of the form $p\; \vec{s}$ where $p\;\vec{t} \approx p\;\vec{s}$.
From this it follows
that $B\ p\; \vec{t} \approx B\ p\; \vec{s}$ and hence that $\Sigma:
B_1',\ldots,B_n' \vdash B_0'$ can be the lower sequent of a rule whose
upper sequent is related in the desired way via permutations to the
upper sequent of the last rule in $\Pi$.

The only remaining rules to consider are $\IL$ and $\CIR$. Once again,
the arguments in these cases are similar and we therefore consider
only the case for $\IL$ in detail. Here, $\Pi$ ends with a rule
of the form 
\begin{equation*}
\infer[\IL]
      {\Sigma : B_1, \ldots, p\; \vec{t} \vdash B_0}
      {\vec{x} : B\; S\; \vec{x} \vdash S\; \vec{x} \qquad
       \Sigma : B_1, \ldots, S\; \vec{t} \vdash B_0}
\end{equation*}
where $p$ is a predicate symbol defined by a clause of the form
$\forall \vec{x}.~p\;\vec{x} \mueq B\ p\;\vec{x}$
and $S$ 
contains no nominal constants. Now, $B_n'$ must be of the form
$p\;\vec{r}$ where $p\;\vec{t} 
\approx p\;\vec{r}$. Noting the proviso on $S$, it follows that $S\;
\vec{t} \approx S\;\vec{r}$. But then the following
\begin{equation*}
\infer[\IL]
      {\Sigma : B_1', \ldots, p\; \vec{r} \vdash B_0'}
      {\vec{x} : B\; S\; \vec{x} \vdash S\; \vec{x} \qquad
       \Sigma : B_1', \ldots, S\; \vec{r} \vdash B_0'}
\end{equation*}
is also an instance of the $\IL$ rule and its upper sequents are
related in the manner needed to those of the $\IL$ rule used in $\Pi$.
\end{proof}

Several rules in \logic require the selection of eigenvariables and
nominal constants. Lemma~\ref{lem:proof-perm} shows that we obtain
what is essentially the same proof regardless of how we choose nominal
constants in such rules so long as the local non-occurrence conditions
are satisfied. A similar observation with regard to the choice of
eigenvariables is also easily verified. We shall therefore identify
below proofs that differ only in the choices of eigenvariables and
nominal constants.

We now turn to the second of our desired results, the preservation of
provability under substitutions.

\begin{lemma}
\label{lem:proof-subst}
Let $\Pi$ be a proof of $\Sigma : \Gamma \vdash C$ and let $\theta$ be
a substitution. Then there is a proof $\Pi'$ of $\Sigma\theta :
\Gamma\cas{\theta} \vdash C\cas{\theta}$ such that $\htf(\Pi') \leq
\htf(\Pi)$. 
\end{lemma}

\begin{proof}
We show how to transform the proof $\Pi$ into a proof $\Pi'$ for the
modified sequent. The transformation is by recursion on $\htf(\Pi)$,
the critical part of it being a consideration of the last rule in
$\Pi$. The transformation is, in fact, straightforward in all cases
other than when this rule is $\unrhdL$, $\forallR$, $\existsL$,
$\existsR$, 
$\forallL$, $\IL$ and $\CIR$. In these cases, we simply apply the
substitution in a nominal capture avoiding way to the lower and any
possible upper sequents of the rule. It is easy to see that the resulting
structure is still an instance of the same rule and its upper sequents
are guaranteed to have proofs (of suitable heights) by induction.

Suppose that the last rule in $\Pi$ is an $\unrhdL$, \ie, it is of the form 
\begin{equation*}
\infer[\unrhdL]{\Sigma : \Gamma, s\unrhd t \vdash C}
      {\left\{\Sigma\rho : \Gamma\cas{\rho} \vdash C\cas{\rho} \;|\;
        \rho\ \mbox{is a solution to}\ s \unrhd t\right\}}
\end{equation*}
Then the following
\begin{equation*}
\infer[\unrhdL]{\Sigma\theta : \Gamma\cas{\theta},
                (s\unrhd t)\cas{\theta} \vdash C\cas{\theta}}
      {\left\{
        \Sigma(\theta\bullet\rho') : \Gamma\cas{\theta \bullet\rho'}
                 \vdash C\cas{\theta \bullet \rho'}
                 \;|\; \rho'\ \mbox{is a solution to}\ (s \unrhd
                t)\cas{\theta}
      \right\}}
\end{equation*}
is also an $\unrhdL$ rule. Noting that if $\rho'$ is a solution to
$(s\unrhd t)\cas{\theta}$, then $\theta\bullet \rho'$ is a solution to
$s\unrhd t$, we see that the upper sequents of this rule are contained
in the upper sequents of the rule in $\Pi$. It follows that we can
construct a proof of the lower sequent whose height is less than or
equal to that of $\Pi$. 

The argument is similar in the cases when the last rule in $\Pi$ is a
$\forallR$ or a $\existsL$ so we consider only the former in
detail. In this case the rule has the form
\begin{equation*}
\infer[\forallR]
      {\Sigma : \Gamma \vdash \forall x.B}
      {\Sigma, h : \Gamma \vdash B[h\;\vec{c}/x]}
\end{equation*}
where $\{\vec{c}\} = \supp(\forall x.B)$. Let $\{\vec{a}\} =
\supp((\forall x. B)\cas{\theta})$. Further, let $h'$ be a new variable
name. We assume without loss of generality that neither $h$ nor $h'$
appear in the domain or range of $\theta$. Letting $\rho = \theta \cup
\{\lambda\vec{c}.h'\;\vec{a}/h\}$, consider the structure
\begin{equation*}
\infer[]
      {\Sigma\theta : \Gamma\cas{\theta} \vdash (\forall x.B)\cas{\theta}}
      {(\Sigma, h)\rho :
               \Gamma\cas{\rho} \vdash B[h\;\vec{c}/x]\cas{\rho}}
\end{equation*}
The upper sequent here is equivalent under $\lambda$-conversion to
$\Sigma\theta, h' : \Gamma\cas{\theta} \vdash (B\cas{\theta})[h'\;
  \vec{a}/x]$ so this structure is, in  fact, also an instance of the
$\forallR$ rule. Moreover, its upper sequent is obtained via
substitution from the upper sequent of the rule in $\Pi$. The lemma
then follows by induction.

The arguments for the cases when the last rule is an $\existsR$ or an
$\forallL$ are similar and so we provide it explicitly only for the
former. In this case, we have the rule
\begin{equation*}
\infer[\existsR]{\Sigma : \Gamma \vdash \exists_\tau x.B}
      {\Sigma, \mathcal{K}, \mathcal{C} \vdash t:\tau &
       \Sigma : \Gamma \vdash B[t/x]}
\end{equation*}
ending $\Pi$. Let $\pi$ be a permutation such that
$\supp(\pi.(B[t/x])) \cap \supp(\theta) = \emptyset$. We assume
without loss of generality that $x$ does not appear in the domain or
range of $\theta$. Then consider the structure
\begin{equation*}
\infer
      {\Sigma\theta : \Gamma\cas{\theta} \vdash
                      (\exists_\tau x.B)\cas{\theta}}
      {\Sigma\theta, \mathcal{K}, \mathcal{C} \vdash
         (\pi.t)[\theta]:\tau &
       \Sigma\theta : \Gamma\cas{\theta} \vdash
         (\pi.B)[\theta][(\pi.t)[\theta]/x]}
\end{equation*}
The typing derivation here is well-formed since permutations and
substitutions are type preserving. Additionally, $\supp(B) \subseteq
\supp(B[t/x])$ implies $\supp(\pi.B)\cap \supp(\theta) = \emptyset$,
and so the conclusion of the lower sequent is equivalent to
$\exists_\tau x. (\pi.B)[\theta]$. Thus this structure is an instance
of the $\existsR$ rule. The term $(\pi.B)[\theta][(\pi.t)[\theta]/x]$
is equal to $(\pi.(B[t/x]))[\theta]$ which is equivalent to
$(B[t/x])\cas{\theta}$. Thus the upper right sequent is obtained via
substitution from the upper right sequent of the rule in $\Pi$. The
lemma then follows by induction.

The only remaining cases for the last rule are $\IL$ and $\CIR$. The
arguments in these cases are, yet again, similar and it suffices to
make only the former explicit. In this case, the end of $\Pi$ has the form
\begin{equation*}
\infer[\IL]
      {\Sigma : \Gamma, p\; \vec{t} \vdash C}
      {\vec{x} : B\; S\; \vec{x} \vdash S\; \vec{x} &
       \Sigma : \Gamma, S\; \vec{t} \vdash C}
\end{equation*}
But then the following
\begin{equation*}
\infer[]
      {\Sigma\theta : \Gamma\cas{\theta}, (p\; \vec{t})\cas{\theta}
                      \vdash C\cas{\theta}}
      {\vec{x} : B\; S\; \vec{x} \vdash S\; \vec{x} &
       \Sigma\theta : \Gamma\cas{\theta}, (S\;
         \vec{t})\cas{\theta} \vdash C\cas{\theta}}
\end{equation*}
is also an instance of the $\IL$ rule. Moreover, the same proof as in
$\Pi$ can be used for the left upper sequent and the right upper
sequent has the requisite form for using the induction hypothesis. 
\end{proof}

The proof of Lemma~\ref{lem:proof-subst} effectively defines a
transformation of a derivation $\Pi$ based on a substitution
$\theta$. We shall use the notation $\Pi\cas{\theta}$ to denote the
transformed derivation. Note that $\htf{(\Pi\cas{\theta})}$ can be less
than $\htf{(\Pi)}$. This may happen because the transformed version of a
$\unrhdL$ rule can have fewer upper sequents.

\begin{cor}
\label{cor:extend}
The following rules are admissible.
\begin{equation*}
\infer[\forallR^*]
      {\Sigma : \Gamma \vdash \forall x.B}
      {\Sigma, h : \Gamma \vdash B[h\;\vec{a}/x]}
\hspace{3cm}
\infer[\existsL^*]
      {\Sigma : \Gamma, \exists x.B \vdash C}
      {\Sigma, h : \Gamma, B[h\;\vec{a}/x] \vdash C}
\end{equation*}
where $h \notin \Sigma$ and $\vec{a}$ is any listing of distinct
nominal constants which contains $\supp(B)$.
\end{cor}
\begin{proof}
Let $\Pi$ be a derivation for $\Gamma \vdash B[h\;\vec{a}/x]$, let $h'$
be a variable that does not appear in $\Pi$, and let $\{\vec{c}\} =
\supp(B)$. By Lemma~\ref{lem:proof-subst}, $\Pi\cas{\lambda\vec{a}.h'\
  \vec{c} / h}$ is a valid derivation. Since $\vec{a}$ contains
$\vec{c}$, no nominal constants appear in the substitution
$\{\lambda\vec{a}.h'\  \vec{c} / h\}$. It can now be seen that the
last sequent in  $\Pi\cas{\lambda\vec{a}.h'\;\vec{c} / h}$ has the
form $\Sigma, h' : \Gamma' \vdash B'$ where $B' \approx
B[h'\;\vec{c}/h]$ and $\Gamma'$ results from replacing 
some of the formulas in $\Gamma$ by ones that they are equivalent to under
$\approx$. But then, 
by Lemma~\ref{lem:proof-perm}, there must be a derivation for $\Sigma,
h' : \Gamma \vdash B[h'\;\vec{c}/h]$. Using a $\forallR$ rule below this we
get a derivation for $\Sigma : \Gamma \vdash \forall x. B$, verifying
the admissibility of $\forallR^*$. The argument for $\existsL^*$ is
analogous.
\end{proof}

We now turn to the main result of this section, the redundancy from a
provability perspective of the $\cut$ rule in \logic. The usual
approach to proving such a property is to define a set of
transformations called cut reductions on derivations that leave the
end sequent unchanged but that have the effect of pushing occurrences
of $\cut$ up the proof tree to the leaves where they can be
immediately eliminated. The difficult part of such a proof is showing
that these cut reductions always terminate. In simpler sequent
calculi such as the one for first-order logic, this argument can be
based on an uncomplicated measure such as the size of the cut formula.
However, the presence of definitions in a logic like \logic renders
this measure inadequate. For example, the following is a natural way
to define a cut reduction between a $\defL$ and a $\defR$ rule that
work on the cut formula: 
\begin{equation*}
\infer[\cut]{\Sigma : \Gamma, \Delta \vdash C}{
  \infer[\defR]
        {\Sigma : \Gamma \vdash p\;\vec{t}}
        {\deduce{\Sigma : \Gamma \vdash B\ p\;\vec{t}}{\Pi'}} &
  \infer[\defL]
        {\Sigma : p\;\vec{t}, \Delta \vdash C}
        {\deduce{\Sigma : B\ p\;\vec{t}, \Delta \vdash C}{\Pi''}}}
\hspace{0.5cm}
\raisebox{1.5ex}{$\Rightarrow$}
\hspace{0.5cm}
\infer[\cut]{\Sigma : \Gamma, \Delta \vdash C}{
  {\deduce{\Sigma : \Gamma \vdash B\ p\;\vec{t}}{\Pi'}} &
  {\deduce{\Sigma : B\ p\;\vec{t}, \Delta \vdash C}{\Pi''}}}
\end{equation*}
Notice that $B\ p\;\vec{t}$, the cut formula in the new cut introduced
by this transformation, could be more complex than $p\;\vec{t}$, the
old cut formula. 
To overcome this difficulty, a more complicated argument based on the
idea of reducibility in the style of Tait \cite{tait67jsl} is often
used. Tiu and Momigliano
\cite{tiu09corr} in fact formulate a notion of parametric
reducibility for derivations that is based on the Girard's proof of
strong normalizability for System F \cite{girard89book} and that works
in the presence of the induction and co-induction rules for
definitions. Our proof makes extensive use of this notion and the
associated argument structure.

\begin{theorem}\label{th:cut-elim}
The $\cut$ rule can be eliminated from \logic without affecting the
provability relation.
\end{theorem}
\begin{proof}
The relationship between \logic and the logic Linc$^-$
treated by Tiu and Momigliano can be understood as follows: Linc$^-$
does not treat the $\nabla$ quantifier and therefore has no rules for
it. Consequently, it does not have nominal constants, it does not
use raising over nominal constants in the rules $\forallR$ and
$\existsL$, it has no need to consider permutations in the $id$ (or
initial) rule and has equality rules in place of nominal abstraction
rules. The rules in \logic other than the ones for $\nabla$, including
the ones for definitions, induction, and co-induction, are essentially 
identical to the ones in Linc$^-$ except for the additional attention
to nominal constants.

Tiu and Momigliano's proof can be extended to \logic in a fairly
direct way since the addition of nominal constants and their
treatment in the rules is quite modular and does not create any new
complexities for the reduction rules. The main issues in realizing this
extension is building in the idea of identity under permutations of
nominal constants and lifting the Linc$^-$ notion of 
substitution on terms, sequents, and derivations to a form that
avoids capture of nominal constants. The machinery for doing this has
already been developed in Lemmas~\ref{lem:proof-perm} and
\ref{lem:proof-subst}. In the rest of this proof we
assume a familiarity with the argument for cut-elimination for Linc$^-$
and discuss only the changes to the cut reductions of Linc$^-$ to
accommodate the differences. 

The $id$ rule in \logic identifies formulas which are equivalent
under $\approx$ which is more permissive than equality under
$\lambda$-convertability that is used in the Linc$^-$ initial
rule. Correspondingly, we have to
be a bit more careful about the cut reductions associated with the
$id$ (initial) rule. For example, consider the following reduction:
\begin{equation*}
\infer[\cut]{\Sigma : B, \Gamma, \Delta \vdash C}{
  \infer[id]
        {\Sigma : \Gamma, B \vdash B'}
        {B \approx B'} &
  \deduce{\Sigma : B', \Delta \vdash C}
         {\Pi'}}
\hspace{1cm}
\raisebox{1.5ex}{$\Rightarrow$}
\hspace{1cm}
\deduce{\Sigma : B', \Delta \vdash C}{\Pi'}
\end{equation*}
This reduction has not preserved the end sequent. However, we know $B
\approx B'$ and so we can now use Lemma~\ref{lem:proof-perm} to
replace $\Pi'$ with a derivation of $\Sigma : B, \Delta \vdash C$.
Then we can use Lemma~\ref{lem:proof-weak} to produce a derivation of
$\Sigma : B, \Gamma, \Delta \vdash C$ as desired. The changes to the
cut reduction when $id$ applies to the right upper sequent of the
$\cut$ rule are similar.

The $\forallR$ and $\existsL$ rules of \logic extend the corresponding
rules of Linc$^-$ by raising over nominal constants in the support of
the quantified formula. The $\forallL$ and $\existsR$ rules of \logic
also extend the corresponding rules in Linc$^-$ by allowing
instantiations which contain nominal constants. Despite these changes,
the cut reductions involving these quantifier rules remain unchanged
for \logic except for the treatment of essential cuts that involve an
interaction between $\forallR$ and $\forallL$ and, similarly, between
$\existsR$ and $\existsL$. The first of these is treated as follows:
\begin{equation*}
\infer[\cut]{\Sigma : \Gamma, \Delta \vdash C}{
  \infer[\forallR]
        {\Sigma : \Gamma \vdash \forall x.B}
        {\deduce{\Sigma, h : \Gamma \vdash B[h\;\vec{c}/x]}{\Pi'}} &
  \infer[\forallL]
        {\Sigma : \Delta, \forall x.B \vdash C}
        {\deduce{\Sigma : \Delta, B[t/x] \vdash C}{\Pi''}}}
\hspace{0.25cm}
\raisebox{1.5ex}{$\Rightarrow$}
\hspace{0.25cm}
\infer[\cut]{\Sigma : \Gamma, \Delta \vdash C}{
  {\deduce{\Sigma : \Gamma \vdash B[t/x]}{\Pi'\cas{\lambda\vec{c}.t/h}}} &
  {\deduce{\Sigma : \Delta, B[t/x] \vdash C}{\Pi''}}}
\end{equation*}
The existence of the derivation $\Pi'\cas{\lambda \vec{c}. t/h}$ (with
height at most that of $\Pi'$) is guaranteed by
Lemma~\ref{lem:proof-subst}. The end sequent of this derivation is
$\Sigma : \Gamma\cas{\lambda \vec{c}. t/h} \vdash B[h\
\vec{c}/x]\cas{\lambda \vec{c}. t/h}$. However,
$\Gamma\cas{\lambda\vec{c}.t/h}\approx \Gamma$ because $h$ is new to
$\Gamma$ and $B[h\;\vec{c}/x]\cas{\lambda \vec{c}. t/h} \approx
B[t/x]$ because $\{\vec{c}\} = \supp(B)$ and so $\lambda \vec{c}. t$ has
no nominal constants in common with $\supp(B)$. Thus, by
Lemma~\ref{lem:proof-perm} and by an abuse of notation, we may
consider $\Pi'\cas{\lambda\vec{c}./h}$ to also be a derivation of
$\Sigma : \Gamma \vdash B[t/x]$. The reduction for a cut involving an
interaction between an $\existsR$ and an $\existsL$ rule is analogous.

The logic \logic extends the equality rules in Linc$^-$ to treat the
more general case of nominal abstraction. Our notion of nominal
capture-avoiding substitution correspondingly generalizes the Linc$^-$
notion of substitution, and we have shown in
Lemma~\ref{lem:proof-subst} that this preserves provability. Thus the
reductions for nominal abstraction are the same as for equality,
except that we use nominal capture-avoiding substitution in place of regular
substitution. For example, the essential cut involving an interaction
between an $\unrhdR$ and an $\unrhdL$ rule is treated as follows:
\begin{equation*}
\infer[\cut]{\Sigma : \Gamma, \Delta \vdash C}{
  \infer[\unrhdR]
        {\Sigma : \Gamma \vdash s\unrhd t}
        {} &
  \infer[\unrhdL]
        {\Sigma : \Delta, s\unrhd t \vdash C}
        {\left\{\raisebox{-1.5ex}{
          \deduce{\Sigma\theta : \Delta\cas{\theta} \vdash C\cas{\theta}}
                 {\Pi_\theta}
         }\right\}}}
\hspace{1cm}
\raisebox{1.5ex}{$\Rightarrow$}
\hspace{1cm}
\deduce{\Sigma : \Delta \vdash C}{\Pi_\epsilon}
\end{equation*}
Here we know $s\unrhd t$ holds and thus $\epsilon$, the identity
substitution, is a solution to
this nominal abstraction. Therefore we have the derivation
$\Pi_\epsilon$ as needed. We can then apply Lemma~\ref{lem:proof-weak}
to weaken this derivation to one for $\Sigma : \Gamma, \Delta \vdash
C$. For the other cuts involving nominal abstraction, we make use of
the fact proved in Lemma~\ref{lem:proof-subst} that nominal capturing
avoiding substitution preserves provability. This allows us to commute
other rules with $\unrhdL$. For example, consider the following occurrence
of a cut where the upper right derivation uses an $\unrhdL$ on a formula
different from the cut formula:
\begin{equation*}
\infer[\cut]{\Sigma : \Gamma, \Delta, s\unrhd t \vdash C}{
  \deduce{\Sigma : \Gamma \vdash B}{\Pi'} &
  \infer[\unrhdL]
        {\Sigma : B, \Delta, s\unrhd t \vdash C}
        {\left\{\raisebox{-1.5ex}{
          \deduce{\Sigma\theta :
                  B\cas{\theta}, \Delta\cas{\theta} \vdash C\cas{\theta}}
                 {\Pi_\theta}
         }\right\}}}
\end{equation*}
Cut reduction produces from this the following derivation:
\begin{equation*}
\infer[\unrhdL]
      {\hspace{2.8cm}\Sigma : \Gamma, \Delta, s\unrhd t \vdash C\hspace{2.8cm}}
      {\left\{\raisebox{-3ex}{
        \infer[\cut]
              {\Sigma\theta : \Gamma\cas{\theta}, \Delta\cas{\theta} \vdash
                C\cas{\theta}}
              {\deduce{\Sigma\theta : \Gamma\cas{\theta} \vdash B\cas{\theta}}
                      {\Pi'\cas{\theta}} &
              \deduce{\Sigma\theta : B\cas{\theta}, \Delta\cas{\theta} \vdash
                C\cas{\theta}}
               {\Pi_\theta}}
       }\right\}\hspace{0.6cm}}
\end{equation*}

Finally, \logic has new rules for treating the $\nabla$-quantifier.
The only reduction rule which deals specifically with either the
$\nablaL$ or $\nablaR$ rule is the essential cut between both rules
which is treated as follows:
\begin{equation*}
\infer[\cut]{\Sigma : \Gamma, \Delta \vdash C}{
  \infer[\nablaR]
        {\Sigma : \Gamma \vdash \nabla x.B}
        {\deduce{\Sigma : \Gamma \vdash B[a/x]}{\Pi'}} &
  \infer[\nablaL]
        {\Sigma : \nabla x.B, \Delta \vdash C}
        {\deduce{\Sigma : B[a/x], \Delta \vdash C}{\Pi''}}}
\hspace{0.25cm}
\raisebox{1.5ex}{$\Rightarrow$}
\hspace{0.25cm}
\infer[\cut]{\Sigma : \Gamma, \Delta \vdash C}
      {\deduce{\Sigma : \Gamma \vdash B[a/x]}{\Pi'} &
       \deduce{\Sigma : B[a/x], \Delta \vdash C}{\Pi''}}.
\end{equation*}

With these changes, the cut-elimination argument for Linc$^-$
extends to \logic, \ie, \logic admits cut-elimination.

\end{proof}

The consistency of \logic is an easy consequence of
Theorem~\ref{th:cut-elim}. 

\begin{cor}\label{consistency}
The logic \logic is consistent, \ie, not all sequents are provable in
it. 
\end{cor}

\begin{proof} The sequent $\vdash \bot$ has no cut-free proof and,
  hence, no proof in \logic.
\end{proof}

The cut-elimination theorem is important for more reasons than showing
the consistency of \logic. As one example, using the cut-rule in
constructing proofs in \logic involves the invention of relevant cut
formulas that function as {\em lemmas}. Thus, knowing that this kind
of creative step is not essential is helpful in designing automatic 
theorem provers that are both practical and complete. 



\section{A Pattern-Based Form for Definitions}
\label{sec:pattern-form}

When presenting a definition for a predicate, it is often convenient
to write this as a collection of clauses whose applicability is also
constrained by patterns appearing in the head. For example, in logics
that support equality but not nominal abstraction, list membership
may be defined by the two pattern based clauses shown below.
\begin{equation*}
\member X (X::L) \triangleq \top \hspace{2cm}
\member X (Y::L) \triangleq \member X L
\end{equation*}
These logics also include rules for directly treating definitions
presented in this way. In understanding these rules, use may be made
of the translation of the extended form of definitions to a version
that does not use patterns in the head and in which there is at most
one clause for each predicate. For example, the definition of the list
membership predicate would be translated to the following form:
\begin{equation*}
\member X K \triangleq (\exists L.~ K = (X :: L)) \lor
 (\exists Y \exists L.~ K = (Y :: L) \land \member X L)
\end{equation*}
The treatment of patterns and multiple clauses can now be understood
in terms of the rules for definitions using a single clause and the
rules for equality, disjunction, and existential quantification.

In the logic \logic, the notion of equality has been generalized to
that of nominal abstraction. This allows us also to expand the
pattern-based form of definitions to use nominal abstraction in
determining the selection of clauses. By doing this, we would allow
the head of a clausal definition to describe not only the term
structure of the arguments, but also to place restrictions on the
occurrences of nominal constants in these arguments. 
For example, suppose we want to describe the contexts in typing
judgments by lists of the form $\tup{c_1,T_1} :: \tup{c_2,T_2} ::
\ldots :: nil$ with the further proviso that each $c_i$ is a distinct
nominal constant. We will allow this to be done by using the following
pattern-based form of definition for the predicate $\cntx$:
\begin{equation*}
\cntx nil \triangleq \top \hspace{2cm}
(\nabla x. \cntx (\tup{x, T} :: L)) \triangleq \cntx L
\end{equation*}
Intuitively, the $\nabla$ quantifier in the head of the second clause
imposes the requirement that, to match it, the argument of $\cntx$
should have the form $\tup{x, T} :: L$ where $x$ is a nominal constant
that does not occur in either $T$ or $L$. To understand this
interpretation, we could think of the earlier definition of {\sl cntx}
as corresponding to the following one that does not use patterns or
multiple clauses:
\begin{equation*}
\cntx K \triangleq (K = nil) \lor
(\exists T \exists L.~ (\lambda x . \tup{x, T} :: L) \unrhd K \land \cntx
L)
\end{equation*}
Our objective in the rest of this section is to develop machinery
for allowing the extended form of definitions to be used directly. We
do this by presenting its 
syntax formally, by describing rules that allow us to employ such
definitions and, finally, by justifying the new rules by means of a
translation of the kind indicated above. 

\begin{defn}
A {\em pattern-based definition} is a finite collection of clauses of
the form
\[\forall \vec{x}.(\nabla \vec{z}. p\ \vec{t}) \triangleq
B\ p\ \vec{x}\] where $\vec{t}$ is a sequence of terms that do not
have occurrences of nominal constants in them, $p$ is a constant such
that $p\ \vec{t}$ is of type $o$ and $B$ is a term devoid of 
occurrences of $p$, $\vec{x}$ and nominal constants and such that 
$B\ p\ \vec{t}$ is of type $o$. Further, we expect such a collection
of clauses 
to satisfy a stratification condition: there must exist an assignment
of levels to predicate symbols such that for any clause $\forall
\vec{x}.(\nabla \vec{z}. p\ \vec{t}) \triangleq B\ p\ \vec{x}$
occurring in the set, assuming $p$ has arity $n$, it is the case that
$\lvl(B\ (\lambda \vec{x}.\top)\ \vec{x}) < \lvl(p)$.
Notice that we allow the collection to contain more than one clause
for any given predicate symbol.
\end{defn}

\begin{figure}[t]
\begin{center}
$\infer[\defR^p]
      {\Sigma : \Gamma \vdash p\; \vec{s}}
      {\Sigma : \Gamma \vdash (B\; p\; \vec{x})[\theta]}$\\[5pt]
for any clause $\forall \vec{x}.(\nabla \vec{z}. p\ \vec{t}) \triangleq
B\ p\ \vec{x}$ in $\cal D$ and any $\theta$\\ such that
$range(\theta)\cap \Sigma = \emptyset$ and $(\lambda \vec{z} . p\ \vec{t})[\theta]
                        \unrhd p\ \vec{s}$ holds\\[20pt]
$\infer[\defL^p]
      {\Sigma : \Gamma, p\; \vec{s} \vdash C}
      {\left\{
        \begin{tabular}{l|l}
         $\Sigma\theta : \Gamma\cas{\theta}, (B\; p\;
         \vec{x})\cas{\theta} \vdash C\cas{\theta}$ &
           $\forall \vec{x}.(\nabla \vec{z}. p\ \vec{t}) \triangleq
                              B\ p\ \vec{x} \in {\cal D}$ and \\
        &
           $\theta$ is a solution to $((\lambda \vec{z} . p\ \vec{t}) \unrhd p\
                        \vec{s})$
       \end{tabular}
              \right\}
      }
$
\end{center}
\caption{Introduction rules for a pattern-based definition $\cal D$}
\label{fig:patterndefrules}
\end{figure}

The logical rules for treating pattern-based definitions are presented
in Figure~\ref{fig:patterndefrules}. These rules encode the
idea of matching an instance of a predicate with the head of a
particular clause and then replacing the predicate with the
corresponding clause body. The kind of matching involved is made
precise through the construction of a nominal abstraction after
replacing the $\nabla$ quantifiers in the head of the clause by
abstractions. The right rule embodies the fact that it is enough if
an instance of any one clause can be used in this way to yield a
successful proof. In this rule, the substitution $\theta$ that results
from the matching must be applied in a nominal capture avoiding way to
the body. However, since $B$ does not contain nominal constants,
the ordinary application of the substitution also suffices. 
To accord with the treatment in the right rule, the left rule
must consider all possible ways in which an instance of an atomic
assumption  $p\ \vec{s}$ can be matched by a clause and must show that
a proof can be constructed in each such case. 

The soundness of these rules is the content of the following theorem
whose proof also makes explicit the intended interpretation of the
pattern-based form of definitions.

\begin{theorem}
The pattern-based form of definitions and the associated proof rules
do not add any new power to the logic. In particular, the $\defL^p$
and $\defR^p$ rules are admissible under the intended interpretation
via translation of the pattern-based form of
definitions. 
\end{theorem}
\begin{proof}
Let $p$ be a predicate whose clauses in the definition being
considered are given by the following set of clauses.
\begin{equation*}
\{\forall \vec{x}_i.~ (\nabla \vec{z}_i. p\ \vec{t}_i) \triangleq
B_i\ p\ \vec{x}_i\}_{i\in 1..n}
\end{equation*}
Let $p'$ be a new constant symbol with the same argument types as
$p$. Then the intended interpretation of the definition of $p$ in a
setting that does not allow the use of patterns in the head and that
limits the number of clauses defining a predicate to one is given by
the clause
\begin{equation*}
\forall \vec{y} . p\ \vec{y} \triangleq \bigvee_{i\in 1..n} \exists \vec{x}_i
. ((\lambda \vec{z}_i . p'\ \vec{t}_i) \unrhd p'\ \vec{y}) \land B_i\
p\ \vec{x}_i
\end{equation*}
in which the variables $\vec{y}$ are chosen such that they do not
appear in the terms $\vec{t}_i$ for $1 \leq i \leq n$. Note also that we are using the term
constructor $p'$ here so as to be able to match the entire 
head of a clause at once, thus ensuring that the $\nabla$-bound
variables in the head are assigned a consistent value for all
arguments of the predicate. 

Based on this translation, we can replace
an instance of $\defR^p$, 
\begin{equation*}
\infer[\defR^p]
      {\Gamma \vdash p\; \vec{s}}
      {\Gamma \vdash (B_i\; p\; \vec{x}_i)[\theta]}
\end{equation*}
with the following sequence of rules, where a double inference line
indicates that a rule is used multiple times.
\begin{equation*}
\infer[\defR]{\Gamma \vdash p\; \vec{s}}
{\infer=[\lorR]
 {\Gamma \vdash \bigvee_{i\in 1..n} \exists \vec{x}_i
  . ((\lambda \vec{z}_i . p'\ \vec{t}_i) \unrhd p'\ \vec{s}) \land
     B_i\ p\ \vec{x}_i
 }
 {\infer=[\existsR]
  {\Gamma \vdash \exists \vec{x}_i
   . ((\lambda \vec{z}_i . p'\ \vec{t}_i) \unrhd p'\ \vec{s}) \land
     B_i\ p\ \vec{x}_i
  }
  {\infer[\landR]
   {\Gamma \vdash ((\lambda \vec{z}_i . p'\ \vec{t}_i)[\theta] \unrhd
     p'\ \vec{s}) \land (B_i\ p\ \vec{x}_i)[\theta]
   }
   {\infer[\unrhdR]{\Gamma \vdash (\lambda \vec{z}_i . p'\
       \vec{t}_i)[\theta] \unrhd p'\ \vec{s}}{}
    &
    \Gamma \vdash (B_i\; p\; \vec{x}_i)[\theta]
   }
  }
 }
}
\end{equation*}
Note that we have made use of the fact that $\theta$ instantiates only
the variables $x_i$ and thus has no effect on $\vec{s}$. Further, the
side condition associated with the $\defR^p$ rule ensures that the
$\unrhdR$ rule that appears as a left leaf in this derivation is
well-formed.

Similarly, we can replace an instance of $\defL^p$,
\begin{equation*}
\infer[\defL^p]
      {\Sigma : \Gamma, p\; \vec{s} \vdash C}
      {\left\{
         \Sigma\theta : \Gamma\cas{\theta}, (B_i\; p\;
         \vec{x}_i)\cas{\theta} \vdash C\cas{\theta}\ |\ 
           \hbox{$\theta$ is a solution to $((\lambda \vec{z} . p\ \vec{t}_i) \unrhd p\
                        \vec{s})$}
              \right\}_{i\in 1..n}
      }
\end{equation*}
with the following sequence of rules
\begin{equation*}
\hspace{-1.8cm}
\infer[\defL]{\Gamma, p\; \vec{s} \vdash C}
{
 \infer=[\lorL]
 {\Gamma, \bigvee_{i\in 1..n} \exists \vec{x}_i
  . ((\lambda \vec{z}_i . p'\ \vec{t}_i) \unrhd p'\ \vec{s}) \land
     B_i\ p\ \vec{x}_i
  \vdash C
 }
 {\hspace{2.5cm}\left\{\raisebox{-6ex}{
  \infer=[\existsL]
   {\Gamma, \exists \vec{x}_i
     . ((\lambda \vec{z}_i . p'\ \vec{t}_i) \unrhd p'\ \vec{s}) \land
     B_i\ p\ \vec{x}_i
    \vdash C
   }
   {\infer[\landL^*]
    {\Gamma, ((\lambda \vec{z}_i . p'\ \vec{t}_i) \unrhd p'\ \vec{s}) \land
     B_i\ p\ \vec{x}_i
     \vdash C
    }
    {\infer[\unrhdL]
     {\Gamma, (\lambda \vec{z}_i . p'\ \vec{t}_i) \unrhd p'\ \vec{s},
       B_i\ p\ \vec{x}_i
       \vdash C
     }
     {\left\{\hbox{
       $\Gamma\cas{\theta}, (B_i\; p\; \vec{x}_i)\cas{\theta}
         \vdash C\cas{\theta}\ |\ \theta$ is a solution to
	   $((\lambda \vec{z} . p'\ \vec{t}_i) \unrhd p'\ \vec{s})$
      }
      \right\}
     }
    }
   }
 }\right\}_{i \in 1..n}\hspace{3.3cm}
 }
}
\end{equation*}
Here $\landL^*$ is an application of $\cL$ followed by $\landL_1$ and
$\landL_2$ on the contracted formula. It is easy to see that the
solutions to $(\lambda \vec{z}.p\;\vec{t}_i) 
\unrhd p\;\vec{s}$ and $(\lambda \vec{z}.p'\;\vec{t}_i)
\unrhd p'\;\vec{s}$ are identical and hence the leaf sequents in this
partial derivation are exactly the same as the upper sequents of the
instance of the $\defL^p$ rule being considered.
\end{proof}

A weak form of a converse to the above theorem also holds. Suppose
that the predicate $p$ is given by the following clauses 
\begin{equation*}
\{\forall \vec{x}_i.~ (\nabla \vec{z}_i. p\ \vec{t}_i) \triangleq
B_i\ p\ \vec{x}_i\}_{i\in 1..n}
\end{equation*}
in a setting that uses pattern-based definitions and that has the 
$\defL^p$ and $\defR^p$ but not the $\defL$ and $\defR$ rules. In such
a logic, it is easy to see that the following is provable:
\begin{align*}
\forall \vec{y} . \left[(p\ \vec{y} \supset \bigvee_{i\in 1..n}\right. \exists&
  \vec{x}_i . ((\lambda \vec{z}_i . p'\ \vec{t}_i) \unrhd p'\ \vec{y})
  \land B_i\ p\ \vec{x}_i) \land~ \\
&\left.(\bigvee_{i\in 1..n} \exists \vec{x}_i 
. ((\lambda \vec{z}_i . p'\ \vec{t}_i) \unrhd p'\ \vec{y}) \land B_i\
p\ \vec{x}_i \supset p\ \vec{y}) \right]
\end{align*}
Thus, in the presence of \cut, the $\defL$ and $\defR$ rules can be
treated as derived rules relative to the translation interpretation
of pattern-based definitions.

We would like also to allow patterns to be used in the heads of
clauses when writing definitions that are intended to pick out the
least and greatest fixed points, respectively. Towards this end we
admit in a definition also clauses of the form $\forall
\vec{x}.(\nabla \vec{z}. p\ \vec{t}) \mueq B\ p\ \vec{x}$ and $\forall
\vec{x}.(\nabla \vec{z}. p\ \vec{t}) \nueq B\ p\ \vec{x}$ with the
earlier provisos on the form of $B$ and $\vec{t}$ and the types of $B$ and
$p$ and with the additional requirement that all the clauses for any
given predicate are unannotated or annotated uniformly with either $\mu$ or
$\nu$. Further, a definition must satisfy stratification conditions as
before. In reasoning about the least or greatest fixed point forms of
definitions, we may use the translation into the earlier, non-pattern
form together with the rules $\IL$ and $\CIR$. It is possible to
formulate an induction rule that works directly from pattern-based
definitions using the idea that to show $S$ to be an induction
invariant for the predicate $p$, one must show that every clause of
$p$ preserves $S$. A rule that is based on this intuition is presented
in Figure~\ref{fig:pattern-induction-rule}. The soundness of this rule
is shown in the following theorem.

\begin{figure}[t]
\begin{equation*}
\infer[\IL^p]
{\Sigma : \Gamma, p\ \vec{s} \vdash C}
{\left\{\vec{x}_i : B_i\ S\ \vec{x}_i \vdash \nabla \vec{z}_i.S\
  \vec{t}_i\right\}_{i\in 1..n} \quad
  \Sigma : \Gamma, S\ \vec{s} \vdash C}
\end{equation*}
\begin{center}
assuming $p$ is defined by the set of clauses $\{\forall \vec{x}_i.
(\nabla \vec{z}_i. p\ \vec{t}_i) \mueq B_i\ p\ \vec{x}_i\}_{i\in 1..n}$
\end{center}
\caption{Induction rule for pattern-based definitions}
\label{fig:pattern-induction-rule}
\end{figure}

\begin{theorem}
The $\IL^p$ rule is admissible under the intended translation of
pattern-based definitions.
\end{theorem}
\begin{proof}
Let the clauses for $p$ in the pattern-based definition be given by
the set
\[\{\forall \vec{x}_i.
(\nabla \vec{z}_i. p\ \vec{t}_i) \mueq B_i\ p\ \vec{x}_i\}_{i\in
  1..n}\]
in which case the translated form of the definition for $p$ would be 
\begin{equation*}
\forall \vec{y} . p\ \vec{y} \mueq \bigvee_{i\in 1..n} \exists \vec{x}_i
. ((\lambda \vec{z}_i . p'\ \vec{t}_i) \unrhd p'\ \vec{y}) \land B_i\
p\ \vec{x}_i.
\end{equation*}
In this context, the rightmost upper sequents of the $\IL^p$ and the $\IL$
rules that are needed to derive a sequent of the form $\Sigma :
\Gamma, p\ \vec{s} \vdash C$ are identical. Thus, to show
that $\IL^p$ rule is admissible, it suffices to show that the left
upper sequent in the $\IL$ rule can be derived in the original
calculus from all but the rightmost upper sequent in an $\IL^p$
rule. Towards this end, we observe that we can construct the following
derivation:
\begin{equation*}
\hspace{-2.7cm}
\infer=[\lorL]
{\vec{y} : \bigvee_{i\in 1..n} \exists \vec{x}_i
  . ((\lambda \vec{z}_i . p'\ \vec{t}_i) \unrhd p'\ \vec{y}) \land
  B_i\ S\ \vec{x}_i
  \vdash S\ \vec{y}
}
{\hspace{2.8cm}\left\{\raisebox{-6ex}{
    \infer=[\existsL]
    {\vec{y} : \exists \vec{x}_i
      . ((\lambda \vec{z}_i . p'\ \vec{t}_i) \unrhd p'\ \vec{y}) \land
      B_i\ S\ \vec{x}_i \vdash S\ \vec{y}
    }
    {\infer[\landL^*]
      {\vec{y}, \vec{x}_i :
        ((\lambda \vec{z}_i . p'\ \vec{t}_i) \unrhd p'\ \vec{y}) \land
        B_i\ p\ \vec{x}_i \vdash S\ \vec{y}
      }
      {\infer[\unrhdL]
        {\vec{y}, \vec{x}_i :
          (\lambda \vec{z}_i . p'\ \vec{t}_i) \unrhd p'\ \vec{y},
          B_i\ S\ \vec{x}_i
          \vdash S\ \vec{y}
        }
        {\left\{\hbox{
            $(\vec{y}, \vec{x}_i)\theta : (B_i\; p\; \vec{x}_i)\cas{\theta}
            \vdash (S\ \vec{y})\cas{\theta}\ |\ \theta$ is a solution to
            $((\lambda \vec{z} . p'\ \vec{t}_i) \unrhd p'\ \vec{y})$
          }
          \right\}
        }
      }
    }
  }\right\}_{i \in 1..n}\hspace{3.5cm}
}
\end{equation*}
Since the variables $\vec{y}$ are distinct and do not occur in
$\vec{t}_i$, the solutions to $(\lambda \vec{z} . p'\ \vec{t}_i)
\unrhd p'\ \vec{y}$ have a simple form. In particular, let $\vec{t}'_i$
be the result of replacing in $\vec{t}_i$ the variables $\vec{z}$ with
distinct nominal constants. Then $\vec{y} = \vec{t}'_i$ will be a most
general solution to the nominal abstraction. Thus the upper sequents
of this derivation will be
\begin{equation*}
\vec{x}_i : B_i\ p\ \vec{x}_i \vdash S\ \vec{t}'_i
\end{equation*}
which are derivable if and only if the sequents
\begin{equation*}
\vec{x}_i : B_i\ p\ \vec{x}_i \vdash \nabla\vec{z}_i. S\ \vec{t}_i
\end{equation*}
are derivable.
\end{proof}

We do not introduce a co-induction rule for pattern-based
definitions largely because we have encountered few interesting
co-inductive definitions that require patterns and multiple clauses.



\section{Examples}
\label{sec:examples}

We now provide some examples to illuminate the properties of nominal
abstraction and its usefulness in both specification and reasoning
tasks; while \logic has many more features, their
characteristics and applications have been exposed in other
work (\eg, see
\cite{momigliano03types,tiu04phd,mcdowell02tocl,tiu10tocl}). In the
examples that are shown, use will be made of the  
pattern-based form of definitions 
described in Section~\ref{sec:pattern-form}. We will also
adopt the convention that tokens given by capital letters denote
variables that are implicitly universally quantified over the entire
clause. 

\subsection{Properties of $\nabla$ and freshness}\label{ssec:fresh-example}

We can use nominal abstraction to gain a better insight into the
behavior of the $\nabla$ quantifier. Towards this end, let the {\sl fresh}
predicate be defined by the following clause.
\begin{equation*}
(\nabla x.\fresh x E) \triangleq \top
\end{equation*}
We have elided the type of {\sl fresh} here; it will have to be
defined at each type that it is needed in the examples we consider
below. Alternatively, we can ``inline'' the definition by using nominal
abstraction directly, \ie, by replacing occurrences of of $\fresh
{t_1} {t_2}$ with $\exists E. (\lambda x. \tup{x, E} \unrhd \tup{t_1,
  t_2})$ for a suitably typed pairing construct $\tup{\cdot,\cdot}$.

Now let $B$ be a formula whose free variables are among $z, x_1,
\ldots, x_n$, and let $\vec{x} = x_1 :: \ldots :: x_n :: nil$ where
$::$ and $nil$ are constructors in the logic.\footnote{We are, once
  again, finessing typing issues here in that the $x_i$ variables may 
  not all be of the same type. However, this problem can be solved by
  surrounding each of them with a constructor that yields a term with
  a uniform type.} Then the following 
formulas are provable from one another in \logic.
\[
\nabla z. B \qquad\quad
\exists z. (\fresh z \vec{x} \land B) \qquad\quad
\forall z. (\fresh z \vec{x} \supset B)
\]
Note that the type of $z$ allows it to be an arbitrary term in the
last two formulas, but its occurrence as the first argument of {\sl
  fresh} will restrict it to being a nominal constant (even when
$\vec{x} = nil$). Figure~\ref{fig:nabla-fresh} shows a derivation for
one of these entailments. Similar proofs can be constructed for the
other entailments.

\begin{figure}[t]
\begin{equation*}
\infer[\nablaL]
  {\nabla z.B \vdash \exists z. (\fresh z \vec{x} \land B)}
  {\infer[\existsR]
     {B[c/z] \vdash \exists z. (\fresh z \vec{x} \land B)}
     {\infer[\landR]
        {B[c/z] \vdash \fresh c \vec{x} \land B[c/z]}
        {
          {\infer[\defR^p]{B[c/z] \vdash \fresh c \vec{x}}{}}
          &
          {\infer[id]{B[c/z] \vdash B[c/z]}{}}
        }
     }
  }
\end{equation*}
\caption{The proof of an entailment involving $\nabla$ and the {\sl
    fresh} predicate}
\label{fig:nabla-fresh}
\end{figure}

In the original presentation of the $\nabla$ quantifier
\cite{miller03lics}, it was shown that one can move a $\nabla$
quantifier inwards over universal and existential quantifiers by using
raising to encode an explicit dependency. To illustrate this, let $B$
be a formula with two variables abstracted out, and let $C \equiv D$
be shorthand for $(C \supset D) \land (D\supset C)$. The following
formulas are provable in the logic.
\begin{align*}
\nabla z. \forall x. (B\ z\ x) &\equiv \forall h. \nabla z.
(B\ z\ (h\ z)) &
\nabla z. \exists x. (B\ z\ x) &\equiv \exists h. \nabla z.
(B\ z\ (h\ z))
\end{align*}
In order to move a $\nabla$ quantifier outwards over universal and
existential quantifiers, one would need a way to make non-dependency
(\ie, freshness) explicit. This is now possible using nominal
abstraction as shown by the following equivalences.
\begin{align*}
\forall x. \nabla z. (B\ z\ x) &\equiv
\nabla z. \forall x. (\fresh z x \supset B\ z\ x)
\\
\exists x. \nabla z. (B\ z\ x) &\equiv
\nabla z. \exists x. (\fresh z x \land B\ z\ x)
\end{align*}
Finally, we note that the two sets of equivalences for moving the
$\nabla$ quantifier interact nicely. Specifically, starting with a
formula like $\nabla z. \forall x. (B\ z\ x)$ we can push the $\nabla$
quantifier inwards and then outwards to obtain $\nabla z. \forall h.
(\fresh z (h\ z) \supset B\ z\ (h\ z))$. Here $\fresh z (h\ z)$ will
only be satisfied if $h$ does not use its first argument, as
expected.

\subsection{Type uniqueness for the simply-typed $\lambda$-calculus}\label{ssec:stlc-example}

\begin{figure*}[t]
\begin{equation*}
\infer{\Gamma \vdash x : a}{x:a \in \Gamma}
\hspace{1cm}
\infer{\Gamma \vdash (t_1\; t_2) : b}
      {\Gamma \vdash t_1 : a \to b & \Gamma \vdash t_2 : a}
\hspace{1cm}
\infer[x \notin \mathop{dom}(\Gamma)]
      {\Gamma \vdash (\lambda x\!:\!a .\; t) : a \to b}
      {\Gamma, x:a \vdash t : b}
\end{equation*}
\caption{Type assignment for $\lambda$-terms}
\label{fig:typing}
\end{figure*}

\begin{figure}[t]
\begin{align*}
& \member P (P::L) \mueq \top \\
& \member P (Q::L) \mueq \member P L \\[10pt]
& \of L X A \mueq \member{\tup{X, A}}{L} \\
& \of L {(\app M N)} B \mueq
\exists A. \of L M (\arr A B) \land \of L N A \\
& \of L {(\tabs A R)} {(\arr A B)} \mueq
\nabla x. \of {(\tup{x, A}::L)} {(R\ x)} B
\end{align*}
\caption{Encoding of type assignment for $\lambda$-terms}
\label{fig:typing-enc1}
\end{figure}

As a more complete example, we consider the problem of showing the
uniqueness of type assignment for the simply-typed
$\lambda$-calculus. The typing rules used in the assignment 
are shown in Figure~\ref{fig:typing}. We introduce the
type $tp$ to denote the collection of simple types and the constants
$i : tp$ to represent the (single) atomic type and $\hbox{\sl
  arr} : tp \to tp \to tp$ to represent the function type
constructor. Representations of $\lambda$-terms will have the type $tm$ 
and will be constructed using the constants $\hbox{\sl app} : tm \to tm \to
tm$ and $\hbox{\sl abs} : ty \to (tm \to tm) \to tm$ that are chosen
to represent application and abstraction, respectively. Finally we
introduce a type $a$ for typing assumptions together with the constant
$\tup{\cdot,\cdot} : tm \to tp \to a$, and the type $alist$ for lists of
typing assumptions constructed from the constants $nil : alist$ and
the infix constant $::$ of type $a \to alist \to alist$.
We define the predicate {\sl member} of type $a \to alist \to o$ and
encode the simple typing of $\lambda$-terms in the definition of a
predicate {\sl of} with type $alist \to tm \to tp \to o$ as shown in
Figure~\ref{fig:typing-enc1}.
Note here that the
side-condition on the rule for typing abstractions is subsumed by the
treatment of $\nabla$ in the logic.

\begin{figure}[t]
\begin{align*}
& \cntx nil \mueq \top \\
(\nabla x. &\cntx (\tup{x,A} :: L)) \mueq \cntx\; L
\end{align*}
\caption{{\sl cntx} in \logic}
\label{fig:cntx}
\end{figure}

\begin{figure}[t]
\begin{align*}
& \cntx nil \mueq \top \\
& \cntx (\tup{X,A} :: L) \mueq (\forall M, N. X = \app
M N \supset \bot) \land\null \\
&\hspace{3.7cm} (\forall R, B . X = \tabs B R
\supset \bot) \land\null \\
&\hspace{3.7cm} (\forall B. \member{\tup{X, B}}{L} \supset \bot) \land \null\\
&\hspace{3.7cm} \cntx L
\end{align*}
\caption{{\sl cntx} in $LG^\omega$}
\label{fig:cntx-lg}
\end{figure}

Given this encoding of simple typing, the task of showing the
uniqueness of type assignment reduces to proving the following formula:
\begin{equation*}
\forall t,a,b. (\of {nil} t a \land \of {nil} t b) \supset a = b.
\end{equation*}
While the theorem that is ultimately of interest is stated with a
$nil$ context, it is not difficult to see that in an inductive proof
we will have to consider the more general case where this context is
not empty. However, the typing context is not entirely arbitrary. It
must have 
the form $\tup{x_1,a_1}::\ldots::\tup{x_n,a_n}::nil$ where each $x_i$
is unique and atomic (a nominal constant). If we assume a predicate
{\sl cntx} which restricts the structure of typing contexts in this
way, then we can state our generalized result as follows.
\begin{equation*}
\forall \ell,t,a,b. (\cntx \ell \land \of \ell t a \land \of \ell t b)
\supset a = b
\end{equation*}
This is now provable by a straightforward induction on either of the
typing assumptions.

We turn now to the question of defining a suitable {\sl cntx}
predicate. Using nominal abstraction, we can define {\sl cntx}
directly and succinctly as shown in Figure~\ref{fig:cntx}. An instance
of the second clause must replace $x$ with a nominal constant and $A$
and $L$ by terms which do not contain that nominal constant. The
atomicity and distinct properties of typing assumptions follow
naturally from this. To better appreciate the elegance of this
approach, consider how one would enforce atomicity and distinctness
without nominal abstraction. In a logic such as $LG^\omega$, the
restrictions imposed by {\sl cntx} would have to be encoded via
negative information as shown in Figure~\ref{fig:cntx-lg}. This
description of typing contexts is cumbersome and non-modular. For
example, if we were to add a new constructor for $\lambda$-terms and a
typing rule associated with this constructor then, even though the
structure of typing contexts has not changed, we would need to change
{\sl cntx} to rule out this constructor from occurring in typing
contexts. We will use the definition of {\sl cntx} with nominal
abstraction going forward.

When proving the generalized type uniqueness property, the typing
context becomes important at two points: when considering the base
case where a typing assumption is looked up in the context, and when
extending the context with a new typing assumption. When a typing
assumption is found in the context, we must show that it is unique.
The definition of {\sl cntx} describes the structure of typing
assumptions that occur at the head of a context, and the following
lemma uses induction to generalize this to arbitrary elements of the
context.
\begin{equation*}
\forall \ell, m, a, b . (\cntx \ell \land \member{\tup{m,a}}\ell
\land \member{\tup{m,b}}\ell) \supset a = b
\end{equation*}
This property can be shown by induction on {\sl cntx} followed by case
analysis on the {\sl member} hypotheses. The interesting case is when
we have $\ell = \tup{m,a}::\ell'$ and $\member {\tup{m,b}} \ell'$.
Applying $\defR^p$ to $\cntx (\tup{m,a}::\ell')$ in this case replaces
$m$ with a nominal constant that $\ell'$ cannot contain. The
assumption that $\member {\tup{m,b}} \ell'$ then leads to a
contradiction, thus eliminating this case. Moving on to the second
point, when adding a typing assumption to the context, we need to show
that the resulting context still satisfies the {\sl cntx} predicate.
This boils down to showing the following.
\begin{equation*}
\forall \ell, a. (\cntx \ell \supset \nabla x. \cntx (\tup{x,a}::\ell))
\end{equation*}
This follows directly from applying $\defR^p$ to {\sl cntx}. With
these issues taken care of, the rest of the type uniqueness proof is
straightforward.

In order for the above reasoning to be meaningful, we must show
that our encoding of the simply-typed $\lambda$-calculus is adequate.
The crux of this is showing that $\Gamma \vdash t : a$ holds in the
simply-typed $\lambda$-calculus if and only if $\vdash \of
{\enc{\Gamma}} {\enc{t}} {\enc{a}}$ is provable in \logic. Here
$\enc{\cdot}$ is a bijective mapping between objects of the
simply-typed $\lambda$-calculus and their representation in \logic.
Since $\logic$ admits cut-elimination, it is straightforward to
analyze how $\vdash \of {\enc{\Gamma}} {\enc{t}} {\enc{a}}$ might be
proved in the logic. Then the only subtlety in showing adequacy is
that the first clause for {\sl of} allows the type of any object to be
looked up in the context while the first typing rule for simply-typed
$\lambda$-calculus only allows the type of variables to be looked up.
This is resolved by noting that typing contexts only contain bindings
for variables. Alternatively, using nominal abstraction, it is
possible to give a definition of typing which is closer to the
original rules (Figure~\ref{fig:typing}) by replacing the first clause
of {\sl of} with the following.
\begin{equation*}
(\nabla x. \of {(L\ x)} x A) \mueq \nabla x. \member{\tup{x,A}}{(L\ x)}
\end{equation*}
An additional benefit of this encoding is that in proofs such as for
type uniqueness we no longer need to consider spurious cases where the
type of a term such as $\app m n$ is looked up in the typing context.

We can now put everything together to establish the type uniqueness
result for the simply-typed $\lambda$-calculus. Suppose $\Gamma \vdash
t : a$ and $\Gamma \vdash t : b$ are judgments in the simply-typed
$\lambda$-calculus. Then by adequacy we know $\vdash \of
{\enc{\Gamma}} {\enc{t}} {\enc{a}}$ and $\vdash \of {\enc{\Gamma}}
{\enc{t}} {\enc{b}}$ are provable in $\logic$. Using these
assumptions, the {\sl cut} rule, and the type uniqueness result proved
earlier in \logic, we know that $\vdash \enc{a} = \enc{b}$ has a proof
in $\logic$. Thus it also has a cut-free proof. This proof must end with
with $\unrhdR$ which means that $\enc{a}$ is equal to $\enc{b}$.
Finally, since $\enc{\cdot}$ is bijective, $a$ must equal $b$.

\subsection{Polymorphic type generalization}

In addition to reasoning, nominal abstraction can also be useful in
providing declarative specifications of computations. We consider the
context of a type inference algorithm that is also discussed in
\cite{cheney08toplas} to illustrate such an application. In this
setting, we might need a predicate {\sl spec} that
relates a polymorphic type $\sigma$, a list of distinct variables
$\vec{\alpha}$ (represented by nominal constants) and a monomorphic
type $\tau$ just in the case that $\sigma =
\forall\vec{\alpha}.\tau$. Using nominal abstraction, we can define
this predicate as follows.
\begin{align*}
&\spec {(\monoTy T)} {nil} T \mueq \top \\
(\nabla x. &\spec {(\polyTy P)} {(x::L)} {(T\ x)}) \mueq
\nabla x. \spec {(P\ x)} L {(T\ x)}.
\end{align*}
Note that we use $\nabla$ in the head of the second clause to
associate the variable $x$ at the head of the list $L$ with its
occurrences in the type $(T\ x)$. We then use $\nabla$ in the body of
this clause to allow for the recursive use of {\sl spec}.

\subsection{Arbitrarily cascading substitutions}
\label{ssect:cascading}

Many reducibility arguments, such as Tait's proof of normalization for
the simply typed $\lambda$-calculus \cite{tait67jsl}, are based on
judgments over closed terms. During reasoning, however, one has often
to work with open terms. To accommodate this requirement, the closed
term judgment is extended to open terms by considering all possible
closed instantiations of the open terms. When reasoning with \logic,
open terms are denoted by terms with nominal constants representing
free variables. The general form of an open term is thus $M\; c_1\;
\cdots\; c_n$, and we want to consider all possible instantiations
$M\; V_1\; \cdots\; V_n$ where the $V_i$ are closed terms. This type
of arbitrary cascading substitutions is difficult to realize in
reasoning systems where variables are given a simple type since $M$
would have an arbitrary number of abstractions but the type of $M$
would {\em a priori} fix that number of abstractions.

We can define arbitrary cascading substitutions in \logic using
nominal abstraction. In particular, we can define a predicate which
holds on a list of pairs $\tup{c_i,V_i}$, a term with the form $M\; c_1\;
\cdots\; c_n$ and a term of the form $M\; V_1\; \cdots\; V_n$. The
idea is to iterate over the list of pairs and for each pair $\tup{c,V}$
use nominal abstraction to abstract $c$ out of the first term and then
substitute $V$ before continuing. The following definition of the
predicate {\sl subst} is based on this idea.
\begin{align*}
& \subst {nil} T T \mueq \top \\
(\nabla x. &\subst {(\tup{x,V}::L)} {(T\; x)} S) \mueq
\subst L {(T\; V)} S
\end{align*}

The ideas in this substitution predicate have been used to formalize
Tait's logical relations argument for the weak normalization of the
simply-typed $\lambda$-calculus in a logic similar to \logic
\cite{gacek08lfmtp}. Here, an important property of arbitrary
cascading substitutions is that they act compositionally. For
instance, taking the slightly simpler example of the untyped
$\lambda$-calculus, we can show that {\sl subst} acts compositionally
via the following lemmas.
\begin{align*}
&\forall \ell, t, r, s.~
\subst \ell {(\app t r)} s \supset
\exists u, v. (s = \app u v \land \subst \ell t u \land \subst \ell r v)
\\
&\forall \ell, t, r.~
\subst \ell {(\abs t)} r \supset
\exists s. (r = \abs s \land \nabla z. \subst \ell {(t\; z)} (s\; z))
\end{align*}
Both of these lemmas have straightforward proofs by induction on {\sl
  subst}.



\section{Related Work}
\label{sec:related}

We structure the discussion of related work into three parts: the
previously existing framework that \logic builds on, alternative
proposals for treating binding in syntax and different approaches for
relating specifications of formal systems and reasoning about them.

\subsection{The precursors for \logic}

The logic \logic that we have described in this paper provides a
framework for intuitionistic reasoning that is characterized by its
use of typed $\lambda$-terms for representing objects, of a
fixed-point notion of definitions with associated principles of
induction and co-induction, of the special $\nabla$-quantifier to
express generic judgments and of nominal abstraction for making
explicit the properties of objects captured by the
$\nabla$-quantifier. All these features except the last derive from
previously described logics. The style in which definitions are
treated originates from work by Schroeder-Heister
\cite{schroeder-heister92nlip} and Girard
\cite{girard92mail}. McDowell and Miller used this idea within a
fragment of the Simple Theory of Types and added to this also a
treatment of induction over natural numbers \cite{mcdowell00tcs}. The
resulting logic, called $\FOLDN$, provides a means for reasoning about
specifications of computations over objects involving abstractions in
which universally quantified judgments are used to capture the dynamic
aspects of such abstractions. While such an encoding suffices for many
purposes, Miller and Tiu discovered its inadequacy in, for example,
treating the distinctness of names in arguments relating to the
$\pi$-calculus and they developed the logic \foldnb with the new
$\nabla$-quantifier as a vehicle for overcoming this deficiency
\cite{miller05tocl}. Tiu then showed how to incorporate inductive and
co-inductive forms of definitions into this context
\cite{tiu04phd}. However, the properties initially assumed for the
$\nabla$-quantifier were too weak to support sophisticated forms of
reasoning based on \hbox{(co-)induction}, and this led to the addition
of the
$\nabla$-strengthening and $\nabla$-exchange principles \cite{tiu06lfmtp}. 
The logic that is a composite of all these features still lacks the
ability, often needed in inductive arguments, to make explicit in a
systematic way properties such as the freshness and
distinctness of nominal constants ({\em i.e.}, the variables bound by
the $\nabla$-quantifier).  Nominal 
abstraction, whose study has been the main focus of this paper,
provides a natural means for reflecting such properties into
definitions and as such represents a culmination of this line of
development.

The exchange property assumed for the $\nabla$-quantifier appears to
have a natural justification. On the other hand, the strengthening
property, while useful in many reasoning contexts, brings with it the
implicit requirement that the types at which $\nabla$-quantifiers are
used be inhabited by an unbounded number of
members. This assumption may complicate the process of showing the
adequacy of an encoding, an important part of using a logical
framework in formalizing the properties of a 
computational system. The observation concerning adequacy has led
Baelde to develop an alternative
approach to enriching the structure provided by \foldnb
\cite{baelde08phd, baelde08lfmtp}. Specifically, he has proposed
treating the $\nabla$-quantifier as a defined symbol, including in its
definition also the ability to lift its predicative effect over
types. The exchange property for the quantifier follows from this
enrichment, while the properties $(\nabla x. P) \supset P$ and $P
\supset (\nabla x.P)$ where $x$ does not occur in $P$ are shown to
hold for certain syntactic classes of formulas.
The resulting
logic has a domain of application that overlaps with that of \logic
but, in our opinion, may not be as convenient to use in actual
reasoning tasks. A
detailed consideration of this issue and also the quantification of
the real differences in adequacy arguments are left for future
investigation. 

\subsection{Nominal logic}
\label{ssec:nominal}

The $\nabla$-quantifier of \logic bears several
similarities to the \new-quantifier contained in nominal logic. As
presented in \cite{pitts03ic}, nominal logic is, in essence, a variant of
first-order logic whose defining characteristics are that it
distinguishes certain domains as those of atoms or names and takes as
primitive a freshness predicate---denoted by the infix operator
$\#$---between atoms and other objects and a swapping operation
involving a pair of names and a term. The logic then formalizes
certain properties of the swapping operation (referred to as
equivariance properties) and of freshness. One of the freshness axioms 
leads to the availability of an unbounded supply of names, an aspect
that is reminiscent of the consequence of the strengthening rule
associated with the $\nabla$-quantifier. Letting $\phi$ be a formula
whose free variables are $a, x_1,\ldots, x_n$ where $a$ is of atom
type, another consequence of the swapping and freshness axioms is the
following equivalence:
\[\exists a. (a \# x_1 \land \ldots \land a \# x_n \land \phi)\quad \equiv\quad
\forall a. (a \# x_1 \land \ldots \land a \# x_n \supset \phi)
\]
The \new-quantifier can be defined in this setting by translating $\new
a. \phi$ into either one of the formulas shown in this equivalence. In
our presentation of \logic, we have taken the $\nabla$-quantifier to be
primitive and we have shown that we can define a {\em fresh} predicate
using nominal abstraction. As we
have seen in Section~\ref{ssec:fresh-example}, we then get a set of
equivalences between $\nabla$, the traditional quantifiers and {\em
  fresh} that is reminiscent of the one discussed here involving the
\new-quantifier.

At a deeper level, there appears to be some convergence in the
treatment of syntax between the nominal logic approach and the one
supported by \logic using $\lambda$-terms. 
For example, both make use of self-dual quantifiers to manage names
and both provide predicates for freshness, equality, and inequality
relating to 
names.  Probably the most fruitful way to compare these approaches in
detail is via their respective proof theories: see
\cite{cheney05fossacs,gabbay07jal} for some proof theory developments
for nominal logics.  To illustrate such a convergence, we note that 
nominal logic has inspired a variant to logic programming in the form
of the $\alpha$Prolog language \cite{cheney08toplas}.  The
specifications written in $\alpha$Prolog have a Horn clause like
structure with the important difference that the \new-quantifier is
permitted to appear in the head.  Clauses of this kind bear a
resemblance to the pattern-based form of definitions discussed in
Section~\ref{sec:pattern-form} in which the $\nabla$-quantifier may
appear at the front of clauses.
In fact, it is shown in \cite{gacek10ppdp} that the former can be directly
translated to the latter.
The animation of such definitions in
\logic through the $\defR^p$ rule requires the solution of nominal
abstraction problems that is similar in several respects to the equivariant
unification \cite{cheney09jar} needed in an interpreter for
$\alpha$Prolog.  

These similarities notwithstanding, the intrinsic structures of nominal
logic and \logic are actually quite different. The former logic is
first-order in spirit and does not include a binding construct at the
outset.  While it is possible to define a (first-order) binding constructor
in nominal logic that obeys the principle of
$\alpha$-equivalence, the resulting binder is not capable of
directly supporting $\lambda$-tree syntax.  In particular,
$\beta$-equivalence is not internalized with these terms: as a
consequence, term-level substitution has to be explicitly formalized
and its formal properties need to be established on a case-by-case basis. 
While such a first-order encoding has some drawbacks from the
perspective of treating binding structure, it also has the benefit
that it can be more easily formalized within the logic
of existing theorem provers such as Coq and Isabelle/HOL 
\cite{pitts06jacm,aydemir06lfmtp,urban08jar}.  

\subsection{Separation of specification and reasoning logics}

An important envisaged use of \logic is in realizing the {\em
two-level approach} to reasoning about the operational semantics of
programming languages and process calculi.  The first step in this
approach is to use a {\it specification logic} to encode such
operational semantics as well as assortments of other properties such
as typing.  The second step involves embedding provability of this
first logic into a second logic, called the {\it reasoning logic}.
This two level-logic approach, pioneered by McDowell and Miller
\cite{mcdowell02tocl,mcdowell97phd}, offers several benefits, such as
the ability to internalize into the reasoning logic properties about
derivations in the specification logic and to use these uniformly in
reasoning about the specifications of particular systems.  For
example, cut-elimination for the specification logic can be used to
prove substitution lemmas in the reasoning logic.  Another benefit is
that $\lambda$-tree syntax is available for both logics since the
specification logic is a simple definition within the reasoning
logic.   Part of our motivation for \logic was for it to play the role
of a powerful reasoning logic.  In particular, 
\logic has been provided an implementation in the Abella system
\cite{gacek08ijcar}.   Given the richer expressiveness of \logic, it
was been possible to redo the example proofs in \cite{mcdowell02tocl}
in a much more understandable way \cite{gacek08ijcar,gacek09phd}.

Pfenning and Sch\"urmann \cite{schurmann98cade} also describe a
two-level approach in which the terms and types of a dependently typed
$\lambda$-calculus called LF \cite{harper93jacm} are used as
specifications and a logic called ${\cal M}_2$ is used for the
reasoning logic.  Sch\"urmann's PhD thesis \cite{schurmann00phd}
further extended that reasoning logic to one called ${\cal M}_2^{+}$.
This framework is realized in the Twelf system \cite{pfenning99cade},
which also provides a related style of meta-reasoning based on mode,
coverage, and termination checking over higher-order judgments in LF.
This approach makes use of $\lambda$-tree syntax at both the
specification and reasoning levels and goes beyond what is available
with \logic in that it exploits the sophistication of dependent types
that also provides for the encoding of proof objects. On the other hand, the
kinds of meta-level theorems that can be proved in this setting are
structurally weaker than those that can be proved in \logic.  For
example, implication and negation are not present in ${\cal M}_2^{+}$
and cannot be encoded in higher-order LF judgments. Concretely, this
means that properties such as bisimulation for CCS or the
$\pi$-calculus are not provable in this approach.

A key component in ${\cal M}_2^{+}$ and in the higher-order LF
judgment approach to meta-reasoning is the ability to specify
invariants related to the structure of meta-logical contexts. These
invariants are called {\it regular worlds} and their analogue in our
system is judgments such as {\sl cntx} which explicitly describe the
structure of contexts. While the approach to proving properties in
Twelf is powerful and convenient for many applications, it may be
preferable to have the ability to define invariants such as {\sl cntx}
explicitly rather than relying on regular worlds, since this allows
more general judgments over contexts to be described, such as in the
example of arbitrary cascading substitutions
(Section~\ref{ssect:cascading}) where the {\sl subst} 
predicate actively manipulates the context of a term.



\section{Acknowledgements}
\label{sec:ack}

We are grateful to Alwen Tiu whose observations with respect to the
earlier formulation of our ideas in terms of $\nabla$-quantifiers in
the heads of clauses eventually led us to a recasting using the
nominal abstraction predicate. Useful comments were also received from
David Baelde and the reviewers of an earlier version of this paper and
our related LICS'08 paper. This work
has been supported by the National Science Foundation grants
CCR-0429572 and CCF-0917140 and by INRIA through the
``Equipes Associ{\'e}es'' Slimmer. Opinions, findings, and conclusions
or recommendations expressed in this papers are those of the authors
and do not necessarily reflect the views of the National Science
Foundation.



\bibliographystyle{elsarticle-num}

\end{document}